\documentclass[12pt]{article}
\usepackage{amsfonts,amsmath}
\usepackage[mathscr]{eucal}
\usepackage{amssymb}
\usepackage{amsthm}
\theoremstyle{plain}
\newtheorem{thm}{Theorem}
\newtheorem{defi}{Definition}
\newtheorem{prop}{Proposition}

\def\theequation{\arabic{section}.\arabic{equation}}
\newcommand{\be}{\begin{eqnarray}}
\newcommand{\ee}{\end{eqnarray}}
\newcommand{\nn}{\nonumber \\}
\newcommand{\lb}{\label}
\newcommand{\p}[1]{(\ref{#1})}

\begin{document}

\begin{titlepage}

\vspace*{0.2cm}

\renewcommand{\thefootnote}{\star}
\begin{center}

{\LARGE\bf  Bi-HKT and bi-K\"ahler supersymmetric sigma  }\\

\vspace{0.5cm}

{\LARGE\bf models.}\\

\vspace{1.5cm}
\renewcommand{\thefootnote}{$\star$}

{\large\bf Sergey~Fedoruk} ${}^\star$,
\quad {\large\bf Andrei~Smilga} ${}^\ast$
 \vspace{0.5cm}

{${}^\star$ \it Bogoliubov Laboratory of Theoretical Physics, JINR,}\\
{\it 141980 Dubna, Moscow region, Russia} \\
{\it On leave of absence from V.N.\,Karazin Kharkov National University, Ukraine} \\
\vspace{0.1cm}

{\tt fedoruk@theor.jinr.ru}\\
\vspace{0.7cm}

{${}^\ast$\it SUBATECH, Universit\'e de Nantes,}\\
{\it 4 rue Alfred Kastler, BP 20722, Nantes 44307, France;}\\
\vspace{0.1cm}

{\tt smilga@subatech.in2p3.fr}\\

\end{center}
\vspace{0.2cm} \vskip 0.6truecm \nopagebreak

   \begin{abstract}
\noindent
We study CKT  (or  bi-HKT)  ${\cal N} = 4$
supersymmetric quantum mechanical sigma models.
They are characterized by the usual and the mirror sectors displaying each HKT geometry.
When the metric involves isometries, a Hamiltonian reduction is possible.
The most natural such reduction
with respect to a half of bosonic target space coordinates produces  an ${\cal N} = 4$ model,
 related
to the  twisted K\"ahler
model due to Gates, Hull and Rocek, but including certain extra $F$-terms in the superfield action.
  \end{abstract}

\vspace{1cm}

\phantom{Keywords: }

\newpage

\end{titlepage}

\setcounter{footnote}{0}

\setcounter{equation}0
\section{Introduction}

The HKT geometries
\footnote{HKT stands for ``hyper-K\"ahler with torsion''.
Though this terminology is probably not the best
(HKT manifolds are {\it not} hyper-K\"ahler and not even K\"ahler),
it is well established in the literature, and we will follow it. }
were first introduced
by physicists in the supersymmetric sigma model framework \cite{HKT}
(see also earlier papers
\cite{IvKrLev,SSTvP,DelVal,Papa95}
where some elements of the HKT structure were displayed)
and  were described
in pure mathematical terms in \cite{Gaud,Grant,Verb}.

For a mathematician, an HKT manifold is a complex manifold of a special kind.
It is endowed with {\it three} different integrable
complex structures satisfying
the quaternion algebra, which are covariantly constant with respect
to one and the same {\it Bismut} torsionful affine {\it connection} \cite{Mavra,Bismut}. When torsions vanish, we are dealing
with hyper-K\"ahler geometry.

One can show that the HKT manifolds are  characterized by a presence of
a closed holomorhic (with respect to any chosen complex structure)
(2,0)--form. This fact can also serve as an alternative definition.

More general geometries --- the so-called CKT (Clifford K\"ahler with torsion) and
OKT (Octonionic K\"ahler with torsion) manifolds were described by physicists in
 \cite{CKT,Hull,ILS,FIS2} and are still awaiting their appreciation by mathematicians.
 Our paper is devoted to the CKT models (or bi-HKT models;  we will see that this latter name
describes more adequately their specifics). We hope that the structures which we display here
may serve as a kind of scaffold to allow for their eventual mathematical description.

In the language used by  physicists, 
HKT supersymmetric quantum mechanical models are described by the
$(\bf{4}, \bf{4}, \bf{0})$
one-dimensional
superfields.
(We follow the notation of \cite{PT} such that the numerals count the numbers of the physical bosonic,
physical fermionic and auxiliary bosonic fields respectively.)
 Generically, one should use {\it nonlinear}  multiplets subject to complicated nonlinear
constraints \cite{DI}. In particular, the
latter are indispensible for describing the hyper-K\"ahlerian models and all the HKT models living on group manifolds. But a wide class of HKT models are described by the linear  $(\bf{4}, \bf{4}, \bf{0})$   multiplets \cite{Copa,IL,IKL}  where the constraints can be explicitly
resolved.

Alternatively, the supersymmetric quantum mechanics (SQM) models can be described in Hamiltonian language. This language is  close to the language used
by mathematicians. For example, in the SQM sigma model involving several
$(\bf{1}, \bf{2}, \bf{1})$ multiplets, the supercharge operator $Q$ is isomorphic to the well-known exterior derivative
operator of the de Rham complex.

The HKT models are characterized by the presence of four real supercharges
$Q_a$ satisfying the standard ${\cal N} = 4$ supersymmetry algebra,
  \be
\lb{SUSY-alg}
 \{Q_a, Q_b \} \ =\ \delta_{ab} H \, .
  \ee
 In recent \cite{QHKT, commentHKT}, an explicit construction of the supercharges of the HKT models was performed.
In \cite{QHKT}, this was done
in real notations, while in \cite{commentHKT}, we presented their complex expressions in terms of holomorphic coordinates.
The complex description involves the closed
holomorphic 2-form mentioned above. It is natural to call the associated antisymmetric tensor
 the {\it hypercomplex structure}.

 It was shown in \cite{taming} that different SQM sigma models are related to each other by Hamiltonian reduction and/or
similarity transformation of the holomorphic supercharges. For example, a Hamiltonian reduction of the Dolbeault model
gives a so-called {\it quasicomplex} de Rham model \cite{quasi}. We found in \cite{commentHKT} that a Hamiltonian reduction
of an HKT model gives a generalized  K\"ahler - de Rham model involving  holomorphic $F$-terms of a particular form
[see Eq.\p{act242-HKT} below]. This model
 represents a special kind of quasicomplex models as is clearly seen in its expression via ${\cal N} =2$ superfields.
In the following, we will call such models {\it quasicomplex K\"ahler models}.

A reader should be careful here. We will see later [Eq.\p{bez-A} and discussion thereafter] that, similar to the story of the HKT manifolds, the actual metric of the manifolds where
quasicomplex K\"ahler models live is {\it not} K\"ahler.   But we were not able to invent a better name. 

 Our new results are the following:

   \begin{enumerate}
   \item We analyzed generic CKT models that involve  three complex structures that are not necessarily quaternionic, but satisfy
the Clifford algebra
 \be
\label{Clifford}
\{I^p, I^q \} = -2\delta^{pq} \, .
  \ee
We have seen that such models involve generically two sectors
characterized each by HKT geometry.  This is related to the known mathematical fact that the enveloping
algebra of three real antisymmetric matrices  satisfying \p{Clifford} is a direct sum of two quaternion algebras
 \cite{Clif=H+H}. Therefore, this kind of models can also be called {\it bi-HKT} models, and that is what we do in the
following.
We have presented simple explicit expressions for two pairs of Hermitian conjugated supercharges.
The first pair represents the Dolbeault supercharges deformed by holomorphic torsions that bring about the nonconservation of the fermion charge \cite{Hull,FIS1}.
  The second pair is related to the first
one by the action of a certain discrete symmetry transformation, an automorphism of the ${\cal N} = 4$ supersymmetry algebra. In constrast to the ordinary HKT
models, the bi-HKT models involve two different hypercomplex structures.

\item The manifolds where the bi-HKT models live have an even complex dimension, their real dimension being an integral multiple of 4. In the case when the metric and the torsions do not depend  on a half of real coordinates (on a half of complex coordinates or on the imaginary parts of all complex coordinates), a Hamiltonian
reduction eliminating these degrees of freedom is possible. One can describe the model thus obtained in terms of the ordinary and mirror  chiral  ${\cal N} = 4$ $(\bf{2}, \bf{4}, \bf{2})$  multiplets. In this way, one obtains
the twisted K\"ahler (we prefer to call it {\it  bi-K\"ahler})
model of Ref. \cite{GHR} generalized by the inclusion  of certain holomorphic
terms in the superfield Lagrangian.

Alternatively, one can describe it in terms of the real ${\cal N} = 2$
({\bf 1}, {\bf 2}, {\bf 1}) superfields.  The reduced model has three distinguishable features:
 {\it i)} it belongs to the class
of  quasicomplex de Rham models with Hermitian (rather than just real) superfield metric,
 {\it ii)} it involves the holomorphic torsions \cite{Braden,torsion,FIS1},
 {\it iii)} it involves two different commuting complex structures whose Bismut torsions have an opposite sign. They take their origin in the two hypercomplex structures of the parent  bi-HKT model.

\end{enumerate}

The plan of the paper is the following. In the next section, we give a  detailed self-contained
 review of the relevant facts of the theory of HKT manifolds.
In Section 3, we discuss generic CKT $\equiv$ bi-HKT models, describe them in the superfield  and also in
the component language. We derive  the compact component expressions for the bi-HKT supercharges which reveal the geometric
structure of these models.
In Section 4, we present and discuss quasicomplex bi-K\"ahler models, derive their component Lagrangian and the supercharges.
We also describe in detail how the Hamiltonian reduction of a generic reducible HKT model and
a generic reducible bi-HKT model is performed. The last section is devoted, as usual, to some concluding remarks. In the Appendix, we present explicit expressions for components Lagrangians in the bi-HKT and bi-K\"ahler models.

\section{HKT sigma models and their Hamiltonian reduction.}

We start with reminding some basic facts of complex geometry.

\begin{defi}
A complex structure for an even-dimensional real manifold is an antisymmetric tensor
$I_{MN}$ satisfying the property  $I_{MN} I^{NK} = - \delta_M^K$ and also the integrability
condition (the vanishing of the Nijenhuis tensor), which can be written in the form
 \be
\lb{Nijen}
 \partial_{[M} I_{N]}{}^P \ =\ I_M{}^Q I_N{}^S \partial_{[Q} I_{S]}{}^P \, .
 \ee
\end{defi}

In the associalted supersymmetric sigma model, this condition allows one to construct
two real supercharges satisfying the algebra \p{SUSY-alg}.
For an integrable complex structure, one can introduce holomorphic coordinates, $x^M = \{z^{\cal M}, \bar z^{\bar {\cal M}} \}$ such that
the metric is Hermitian (one can always do it locally, but a nontrivial  property following from \p{Nijen}
is that the manifold can be divided
into a set of overlapping holomorphic charts with holomorphic glue functions),
 \be
\lb{def-h}
ds^2 = \ 2h_{{\cal M} \bar {\cal N}} dz^{\cal M}  d{\bar z}^{\bar {\cal N}}\, .
 \ee

In these coordinates,
the tensor $I_M{}^N$ has the following
nonzero components,
  \be
\lb{Icompl}
I_{\cal M}{}^{ {\cal N}} = -I^{\cal N}_{\ {\cal M}} = -i \delta_{\cal M}^{\cal N}, \qquad I_{\bar {\cal M}}{}^{\bar {\cal N}} = -I^{\bar {\cal N}}_{\ \bar {\cal M}} = i
\delta_{\bar {\cal M}}^{\bar {\cal N}} \, .
  \ee
It follows that $I_{{\cal M} {\bar {\cal N}}} = - I_{{\bar {\cal N}} {\cal M}} = -ih_{{\cal M}{\bar {\cal N}}}$.

\vspace{.2cm}

\begin{defi}
A K\"ahler manifold is a complex manifold  where $I_{MN}$ is covariantly constant,

  \be
\label{Lev-Civ}
\nabla_P I_{MN} = \ \partial_P I_{MN} - \Gamma^Q_{PM} I_{QN} -
\Gamma^Q_{PN} I_{MQ} \ =\  0 \, .
  \ee
\end{defi}

 It follows that the K\"ahler form $\Omega = I_{MN} \, dx^M \wedge dx^N$ is closed,
$d\Omega = 0$.

\vspace{.2cm}

For a generic complex manifold, the complex structure is not covariantly constant with respect to the usual Levi-Civita
connection defined in \p{Lev-Civ}, but {\it torsionful} connections
(where the Christoffel symbols are modified to $  \Gamma^Q_{PM} =  \Gamma^Q_{PM} + \frac 12\, g^{QS} C_{S[PM]}$)
satisfying the conditions
$\tilde {\nabla}_P g_{MN} = \tilde {\nabla}_P I_{MN} \ = 0$ exist.
If we require for the torsion tensor $C_{SPM}$  to be totally antisymmetric, there is only one such connection, called  Bismut
 connection \cite{Mavra,Bismut}.

Explicitly,
\be
\lb{CBismut}
 C^{(Bismut)}_{MNK}(I) \ =\  I_M{}^P I_N{}^Q I_K{}^R \Big(\nabla_P I_{QR} +  \nabla_Q I_{RP} +
\nabla_R I_{PQ}\Big) \, .
\ee

Another distinguished connection is the  Obata connection \cite{Obata}
(see \cite{Soldat} for a pedagogical discussion).
It is torsionless. The Obata covariant derivative
of the complex structure vanishes, but the metric in this case is not covariantly constant.
In all the models considered in this paper and based on the linear $(\bf{4}, \bf{4}, \bf{0})$ multiplets, the Obata curvature vanishes. One can  {\it conjecture} that the inverse is true and
 that all Obata-flat HKT manifolds can be described with linear $(\bf{4}, \bf{4}, \bf{0})$ multiplets.

\begin{defi}
 A hyper-K\"ahler manifold is a manifold with three different integrable
complex structures $I^p$ that satisfy the quaternion algebra
   \be
  \label{quatern}
   I^p I^q = -\delta^{pq} + i \epsilon^{pqr} I^r
    \ee
and are covariantly constant in a usual way \p{Lev-Civ}.
\end{defi}

\begin{prop}
The real antisymmetric matrices satisfying \p{quatern}
should have dimension $4n^*$ with integer $n^*$. Locally, one can always choose a basis where they
acquire the canonical form,
 \be
\label{Ia_canon_HKT}
(I^1)_M{}^N \ =\ {\rm diag} (\EuScript{I},\ldots, \EuScript{I}), \ \ \ \
(I^2)_M{}^N \ =\ {\rm diag} (\EuScript{J},\ldots, \EuScript{J}), \ \ \ \
(I^3)_M{}^N \ =\ {\rm diag} (\EuScript{K},\ldots, \EuScript{K})  ,
 \ee
where
 \be
\lb{IJKcanon}
\!\!
\EuScript{I} \  = \left( \begin{array}{cccc} 0&-1&0&0 \\ 1&0&0&0 \\ 0&0&0&-1 \\ 0&0&1&0 \end{array} \right),\quad
\EuScript{J} \  =  \left( \begin{array}{cccc} 0&0&1&0 \\ 0&0&0&-1 \\ -1&0&0&0 \\ 0&1&0&0 \end{array} \right)
 ,\quad
\EuScript{K} \  = \left( \begin{array}{cccc} 0&0&0&1 \\ 0&0&1&0 \\ 0&-1&0&0 \\ -1&0&0&0 \end{array} \right)  .
\ee
\end{prop}

The sign choice corresponds to the convention
  \be
\lb{sign-z}
z^{\cal M} = x^{{\cal M}1} + i x^{{\cal M}2}\, , \ \ \ \ \ \ \ \ \
\partial_{\cal M} = \frac 12  \left( \frac \partial{\partial x^{{\cal M}1}} - i \frac \partial{\partial x^{{\cal M}2}}  \right) \, ,
   \ee
 which we follow in the most though not in all the cases.

The $4 \times 4$ matrices $\EuScript{I},\EuScript{J},\EuScript{K}$ are self-dual. They are related to the `t Hooft symbols.

\begin{defi}
An HKT manifold is a manifold with three quaternionic complex structures which  are covariantly constant with {\it one and the same}
 torsionful Bismut affine connection,
 \be
\lb{CICJCK}
  C^{(B)}_{MNK}(I^1) \ =\  C^{(B)}_{MNK}(I^2) \ =\  C^{(B)}_{MNK}(I^3) \, .
 \ee
\end{defi}

Both for the hyper-K\"ahler manifolds and for the HKT manifolds, one can construct ${\cal N}=4$ supersymmetric sigma models involving the
$(\bf{4}, \bf{4}, \bf{0})$ multiplets. For hyper-K\"ahler models one can also construct ${\cal N}=8$ supersymmetric sigma models involving the
$(\bf{4}, \bf{8}, \bf{4})$ multiplets \cite{HK}.

 Consider the forms $\Omega_p$ associated with each complex structure.

 \begin{prop}
For an HKT manifold, the form
  \be
\label{Ical}
{\cal I} = \Omega_2 + i\Omega_3
 \ee
 has the type (2,0) with respect to the complex structure $I^1$. Its exterior holomorphic derivative
vanishes, $\partial_1 {\cal I} = 0$.
\end{prop}

It follows that the complex conjugated form $\bar {\cal I} = \Omega_2 + i\Omega_3 $ has the type (0,2) with respect to $I^1$.
By symmetry, it is also true, of course,  that the forms $\Omega_3 \pm i\Omega_1$ have the types (2,0)  and (0,2) with respect to the complex structure $I^2$
and the forms $\Omega_1 \pm i\Omega_2$ have the types (2,0) and (0,2) with respect to $I^3$.

The {\it proof} of this important statement is given in \cite{Gaud,Grant} (see \cite{commentHKT} for  pedagogical explanations).
One can show that it works both ways such that the existence of a closed holomorphic form can be chosen as an alternative
definition of an HKT geometry. As follows from {\bf Proposition 1},
a complex basis can be chosen such that
 \be
\lb{I-canon}
{\cal I}_{\cal M}{}^{\bar {\cal N}} \ = h^{\bar {\cal N}  {\cal P} }{\cal I}_{{\cal M} {\cal P}}   \ = {\rm diag}  (\epsilon, \ldots, \epsilon) \,, \quad \quad \quad \quad
\epsilon = \left( \begin{array}{cc} 0 & 1 \\ -1 & 0 \end{array} \right)  \, ,
 \ee
and the same for $\bar {\cal I}_{\bar {\cal M}}{}^{{\cal N}} $.

Consider the following two pairs of complex conjugated supercharges,
\begin{equation}\label{SR-HKT-s}
\begin{array}{rcl}
S^{\,{}^{\rm HKT}}&=& \sqrt{2} \, \psi^{\cal M}\left[\Pi_{\cal M}-
{ \frac{i}{2}}\left(\partial_{{\cal M}}h_{{\cal K}\bar {\cal L}}\right) \psi^{\cal K}\bar\psi^{\bar {\cal L}}\right],\\[8pt]
\bar S^{\,{}^{\rm HKT}}&=& \sqrt{2} \, \bar\psi^{\bar {\cal M}}
\left[\bar \Pi_{\bar {\cal M}}+{ \frac{i}{2}}\left(\bar\partial_{\bar {\cal M}}h_{{\cal K}\bar {\cal L}}\right)
\psi^{\cal K}\bar\psi^{\bar {\cal L}}\right], \\[8pt]
R^{\,{}^{\rm HKT}}&=& \sqrt{2} \, \psi^{\cal N}\, {\cal {\cal I}}_{\,{\cal N}}{}^{\bar {\cal M}}\left[\bar \Pi_{\bar {\cal M}}-
{ \frac{i}{2}}\left(\bar\partial_{\bar {\cal M}}h_{{\cal K}\bar {\cal L}}\right)
\psi^{\cal K}\bar\psi^{\bar {\cal L}}\right],\\[8pt]
\bar R^{\,{}^{\rm HKT}}&=& \sqrt{2} \, \bar\psi^{\bar {\cal N}}\, \bar {\cal I}_{\,\bar {\cal N}}{}^{{\cal M}}\left[\Pi_{\cal M}+
{ \frac{i}{2}}\left(\partial_{{\cal M}}h_{{\cal K}\bar {\cal L}}\right)\psi^{\cal K}\bar\psi^{\bar {\cal L}}\right] \, ,
\end{array}
\end{equation}
where  $\psi^{\cal M}, \bar \psi^{\bar {\cal N}}$ are fermionic Grassmann variables
 carrying {\it world} indices and $\Pi_{\cal M}, \bar\Pi_{\bar {\cal M}}$ are the canonical momenta,
obtained from the component Lagrangian by variation over $\dot{z}^{\cal M}, \ \dot{\bar z}^{\bar {\cal M}}$
while keeping  $\psi^{\cal N}$ and       $ \bar \psi^{\bar {\cal N}}$  fixed. 
 Their Poisson brackets are nontrivial,
 \be
\label{Poisferm}
\{z^{\cal M}, \Pi_{\cal N} \} \ =\ \delta_{\cal N}^{\cal M} \, ,  \ \ \
\{\bar \psi^{\bar {\cal N}},  \psi^{{\cal M}} \} = \ -i h^{\bar {\cal N} {\cal M}} \, ,\ \ \ \{ \Pi_{\cal M}, \psi^{\cal N} \} = - \frac 12 \partial_{\cal M}
h^{\bar {\cal K}  {\cal N}} \psi_{\bar {\cal K}} \, , \ \ \ {\rm etc.}
 \ee
The matrix ${\cal I}$ was defined in \p{Ical}.

\begin{prop}
For a HKT manifold, the supercharges \p{SR-HKT-s} satisfy the standard ${\cal N} = 4$ supersymmetry algebra,
\be
\label{algQR}
\begin{array}{l}
\{Q,Q\} =  \{R,R\} = \{ Q, R\} = \{ Q, \bar R\} = 0 \ \ \ \ \ \ \ \ \ \ \ {\rm and \ c.c.}\ ,  \\ [7pt]
\{Q,\bar Q\} = \{R, \bar R\} = H \, .
\end{array}
\ee
  \end{prop}

{\it Proof} (see Ref. \cite{commentHKT}).
One can verify the validity of the algebra \p{algQR} for the supercharges
 \p{SR-HKT-s} explicitly. Alternatively,
the latter  can be derived as N\"other integrals of motion,
using the Lagrangian formalism.

\vspace{.2cm}

Let us make now two remarks.
 \begin{enumerate}
\item
The expressions \p{SR-HKT-s} are very much similar
to the expressions of the supercharges in the standard K\"ahler - de Rham model
 (the SQM version of the  K\"ahler sigma model in two spacetime dimensions introduced in
\cite{Zumino}). The latter have the form  \cite{taming,commentHKT},
\begin{equation} \label{QdeRham}
\begin{array}{rcl}
Q&=&\psi^M\left(\Pi_M-\frac i2\, \partial_M g_{NP} \, \psi^N\bar\psi^P\right),
\\[8pt]
\bar Q&=&\bar\psi^M\left(\Pi_M+\frac i2\, \partial_M g_{NP} \,
\psi^P\bar\psi^N\right),
\\[8pt]
R &=&\psi^S I_S^{\ M} \left(\Pi_M-\frac i2\, \partial_M g_{NP} \, \psi^N\bar\psi^P\right),
\\[8pt]
\bar R &=&\bar\psi^S I_S^{\ M} \left(\Pi_M+\frac i2\, \partial_M g_{NP} \,
\psi^P\bar\psi^N\right),
\end{array}
\end{equation}
where $I$ is the complex structure matrix. Bearing this in mind, it is natural to call the matrix ${\cal I}$
entering \p{SR-HKT-s} the matrix
of {\it hypercomplex structure}.

\item The supercharges written above represent classical functions defined on the phase space of the system.
One can also construct quantum supercharges, which are
 operators acting in the Hilbert space of the SQM system that is isomorphic
to the space of forms. When going from classical function to quantum operators,
one should resolve ordering ambiguities. The general
recipe for this was given in \cite{howto} --- the {\it covariant} quantum  supercharges
(acting on  the Hilbert space where the inner product
is defined with the covariant integration measure) are obtained from their classical counterparts by Weyl
ordering and a subsequent conjugation,
 \be
\label{conjug}
Q^{\rm cov} \ =\ (\det g)^{-1/4} Q^{\rm Weyl\ ordered}  (\det g)^{1/4}
 \ee
 Bearing in mind  the regularity of this procedure, we will  talk very little
in this paper about the quantum supercharges and mostly discusss the classical ones.
The symbol $\{ \cdot, \cdot \}$  will thus mean the Poisson brackets rather than anticommutators.

\end{enumerate}

 We concentrate on the  HKT models described by the
superfield Lagrangian  expressed  into several linear ({\bf 4}, {\bf 4}, {\bf 0}) multiplets, ${\cal V}^{i \alpha}_a$,
$a = 1, \ldots , n^*$ being the flavor index.
A ({\bf 4},\,{\bf 4},\,{\bf 0}) multiplet lives in the ${\cal N} = 4$ superspace having the coordinates
$(t,\theta^{i{k}^\prime})$,
with $\theta^{i{k}^\prime}$ satisfying the pseudoreality condition,
  \be
\lb{pseudo}
(\overline{\theta^{i{k}^\prime}})= -\epsilon_{ij}\epsilon_{{k}^\prime{l}^\prime}\theta^{j{l}^\prime}
\equiv -\theta_{ik'} \, .
  \ee
The indices $i=1,2$ and $k^\prime=1,2$ are doublet indices of the ${\rm SU}_L(2)$ and ${\rm SU}_R(2)$ groups, respectively. The latter
form together the full automorphism group ${\rm SO}(4)={\rm SU}_L(2)\times{\rm SU}_R(2)$ of
the ${\cal N}{=}\,4$ superalgebra. Each multiplet carries two spinor indices; the Latin index is
transformed by the group ${\rm SU}_L(2)$, and the Greek one by the extra internal  symmetry
$SU(2)$ group (usually called  {\it Pauli-G\"
ursey group} \cite{PG}) that commutes with the supersymmetry
transformations.

Like $\theta^{i k'}$, $ {\cal V}^{i \alpha}$ is pseudoreal,
$$(\overline{{\cal V}^{i\alpha}})=
-\epsilon_{ij}\epsilon_{\alpha \beta}{\cal V}^{j\beta}  $$
and satisfies the constraints
    \begin{equation}\label{const1}
D^{(i{j}^\prime}{\cal V}^{k)\alpha} = 0\, ,
     \end{equation}
  where
\begin{equation}\label{D}
D^{i{k}^\prime} = \frac{\partial}{\partial\theta_{i{k}^\prime }}+i\theta^{i{k}^\prime}\partial_t
\end{equation}
are the covariant derivatives.

It is convenient then to introduce
two complex superfields ${\cal V}^m$ such that
  \be
\lb{V-compl}
{\cal V}^{i\alpha} \ =\  \left( \begin{array}{cc}  {\cal V}^1 &   {\cal V}^2 \\  \bar {\cal V}^2
&  -\bar {\cal V}^1
\end{array} \right)
 \ee
and also
\be
\label{thetet-def}
\theta^{i k'} \ =\ \left( \begin{array}{rr}  \eta &  \theta \\ \bar \theta &
-\bar \eta
\end{array} \right) \, , \ \ \ \ \ \ \ \ \ \ D^{ik'}
 \ =\ \left( \begin{array}{rr} \bar D_\eta &  \bar D_\theta \\
-D_\theta & D_\eta
\end{array} \right) \, ,
 \ee
with the convention
  \be
\label{D_thetet}
\begin{array}{c}
D_\theta = \partial_\theta - i \bar \theta \partial_t \,, \qquad
  \bar D_\theta = -\partial_{\bar \theta} + i  \theta \partial_t \, ,  \\ [7pt]
D_\eta = \partial_\eta - i \bar \eta \partial_t \,, \qquad
\bar D_\eta = -\partial_{\bar \eta} + i  \eta \partial_t \, .
\end{array}
\ee
The constraints \p{const1} can now be rewritten in a nice form

 \be
\label{dop-eq2}
\begin{array}{c}
\bar D_\theta {\cal V}^m = \bar D_\eta {\cal V}^m = 0\,, \qquad
 D_\theta \bar {\cal V}^{\bar m} =  D_\eta \bar {\cal V}^{\bar m} = 0\,,
\\ [7pt]
D_\theta {\cal V}^{m}+\epsilon^{mn}\bar D_\eta \bar{\cal V}^{\bar{n}}=0 \,,\qquad
D_\eta {\cal V}^{m}-\epsilon^{mn}\bar D_\theta \bar{\cal V}^{\bar{n}}=0\, .
\end{array}
 \ee
(with $\epsilon^{12} = -\epsilon_{12} = -1$).

Alternatively, one can represent  $ {\cal V}^{i \alpha}$ via real 4-vector  components as
 \be
\label{Vvect}
{\cal V}^{i\alpha} \ =\   {\cal V}^{M = 1,2,3,4} (\sigma_M^\dagger)^{i\alpha}
 \ =\
\left( \begin{array}{cc}  {\cal V}^3- i{\cal V}^4 &  {\cal V}^1 -i {\cal V}^2 \\   {\cal V}^1 + i {\cal V}^2
&  -{\cal V}^3- i{\cal V}^4
\end{array} \right)
\ee
with $\sigma_M^\dagger = (\vec{\sigma}, -i)$. (Obviously, it is one of many possible choices.) Then
 \be
\lb{v12}
 {\cal V}^{m=1} \ =\ {\cal V}^{M=3}- i{\cal V}^{M=4}, \ \ \ \ \ \ \ \ \ \ {\cal V}^{m=2} \ =\
{\cal V}^{M=1}- i{\cal V}^{M=2} \, .
 \ee

The constraints \p{dop-eq2} can be easily resolved such that ${\cal V}^m$ is expressed via a couple of standard ${\cal N} =2$ chiral superfields
$V^m = \ v^m + \sqrt{2} \theta \psi^m - i \theta \bar \theta \dot{v}^m$,
   \begin{eqnarray}\label{V-4-2exp}
&{\cal V}^{m}=V^{m}+\eta\, \epsilon^{mn}\bar D\bar V^{\bar n} -i\eta\bar\eta \,\dot{V}^{m} \,,\quad &
\bar{\cal V}^{\bar{m}}=\left({\cal V}^{m}\right)^\dagger= \bar V^{\bar{m}}-\bar \eta\, \epsilon^{{\bar m}\bar{n}} D V^{n}
+i\eta\bar\eta \,\dot{\bar V}^{\bar{m}}
\end{eqnarray}
(with $D \equiv D_\theta, \  \bar D \equiv \bar D_\theta$).

\vspace{1mm}

A generic  HKT action is expressed in these terms as
   \be
\label{act1}
 S \ =\ { \frac 14} \int dt \, d \theta d\bar \theta d\eta d\bar \eta\,
 {\cal L}({\cal V}^m_a, \bar {\cal V}^{\bar m}_a) \, .
   \ee
After integrating over $d\eta d\bar \eta$, it acquires the form ,
   \be
\lb{act1-N2}
S={ \frac14}\displaystyle{\int}dt\, d\theta d\bar\theta\,  \Delta^{ab}_{m\bar n}{\cal L}(V, \bar V) \, D V^{m}_{a}\bar D
\bar V^{\bar n}_{b} \, ,
   \ee
where
\be
\lb{metr-N4}
 \Delta_{m\bar n}^{ab}{\cal L} \ =\  \Big(\partial^a_m \bar \partial_{\bar n}^{b}    +
\epsilon_{m k} \epsilon_{\bar n \bar l} \bar \partial^a_{\bar k}  \partial^b_l \Big) {\cal L}
 \ee
is a Hermitian metric $h_{m\bar n}$.
The Lagrangian \p{act1-N2} belongs to the wider class of ${\cal N} =2$ SQM sigma models describing the Dolbeault complex
\cite{IS}.

 \begin{defi}
We will call the Dolbeault SQM model {\it reducible} if its metric  $h_{{\cal M}\bar {\cal N}}$    does not depend on a half of real coordinates.
\end{defi}

Indeed,
the corresponding canonical
momenta commute in this case with the Hamiltonian, and one can perform the Hamiltonian reduction.
Many different schemes of the Hamiltonian reduction are possible \cite{taming},
but in this paper we only consider the reduction  killing a half of bosonic dynamical variables.

In our problem, it will sometimes be convenient to consider the reduction with respect to a half of {\it complex} coordinates $z^{\cal M}$  (${\cal M} \equiv \{m,a\}$)
 and sometimes with respect to the imaginary parts
Im($z^{\cal M}$). In the latter case,  the Hermitian metric
$h_{{\cal M} \bar {\cal N}}$ goes
over to the sum
 \be
\label{metr-red}
h_{{\cal M} \bar {\cal N}} \ \to \ \frac 12\, \Big( g_{(MN)} + i b_{[MN]} \Big) \, ,
 \ee
involving a real symmetric and an imaginary antisymmetric part. The factor $1/2$ in \p{metr-red} corresponds to the factor 2 in \p{def-h}. The convenience of this convention will be further clarified in Sect. 4 when discussing the reduction of the HKT and bi-HKT models.

  \begin{prop}
A generic reducible Dolbeault SQM model gives after  reduction a  {\it quasicomplex} de Rham model
where the metric tensor is replaced in many formulas of the ordinary de Rham model by the Hermitian sum \p{metr-red}. The actual real
metric tensor of the reduced
model is given by
  \be
G_{MN} \ \to \  g_{MN} + b_{MK} (g^{-1})^{KL} b_{LN} \, .
 \ee
  \end{prop}
See \cite{quasi} for the {\it proof} and discussion.

\begin{prop}
A reducible HKT SQM model gives after reduction
 a quasicomplex model of a special form.
It is expressed into $n^*$ ({\bf 2}, {\bf 4}, {\bf 2}) superfields ${\cal Z}_a (t; \theta, \eta;  \bar\theta, \bar\eta)$
satisfying the linear chiral constraints $ \bar D_\theta {\cal Z}^a  =  \bar D_\eta {\cal Z}^a\ =\ 0 \, $.
 The superfield Lagrangian describes a quasicomplex K\"ahler  model: a K\"ahler model generalized by including
 certain holomorphic $F$-terms,
     \be
\lb{act242-HKT}
 S =  { \frac14}\displaystyle{\int}dt\, d^2\theta \,
 d^2 \eta\, {\cal K}({\cal Z},\bar{\cal Z}) -
  { \frac12}  \left[ \displaystyle{\int}dt\, d\theta d\eta \,
{\cal A}_{a} ({\cal Z}) \dot{\cal Z}^{a}  + {\rm c.c.} \right] \, .
    \ee
\end{prop}

This statement was made in  \cite{commentHKT}. We will give its rigourous proof in
Sect. 4.1 before discussing the reduction of the bi-HKT models.

\section{Bi-HKT manifolds. }
\setcounter{equation}0

 \begin{defi}
CKT manifolds are manifolds involving three integrable complex structures
that satisfy the Clifford
algebra \p{Clifford}, but not the quaternion algebra  \p{quatern}. They also  involve a totally antisymmetric
torsion tensor $C_{MNL}$
that satisfies the following condition: for each complex structure $I^p$, the tensor  $C_{MNL}$  represents
a sum $C_{MNL} = B^p_{MNL} + H^p_{MNL}$, where  $ B^p_{MNL}$ is the torsion
\p{CBismut} of the Bismut
connection for the complex structure $I^p$ (these connections are {\it not} necessarily the same),
  and  the form corresponding to the tensor $H^p_{MNL}$
is the sum of an exact holomorphic [of the type (3,0) ] and the conjugated
antiholomorphic [of the type (0,3) ] forms with respect
to $I^p$.

\end{defi}

This definition given in \cite{FIS2} is equivalent to the original definition of \cite{CKT} (formulated in a somewhat
more complicated way).

In contrast to the HKT case with three quaternionic complex structures, the algebra of the matrices $A + B^p I^p$ is not
closed under
multiplication.

  \begin{prop}
 The closure of this algebra represents a direct sum of two quaternion algebras ${\mathcal H}_+ +
{\mathcal H}_-$ \cite{Clif=H+H}.
\end{prop}

       \begin{proof}
 The products  of the original generators $I^p$ give 4 extra generators,
\be
 J^p = \frac 12\, \epsilon^{pqr} I^q I^r, \ \ \ \ \ \Delta = - I^1 J^1 = -I^2 J^2 = -I^3 J^3 \, .
 \ee
The algebra now closes,
 \be
\lb{alg-CKT}
\begin{array}{ccc}
&& I^p J^q = J^p I^q = \ -\delta^{pq} \Delta +  \epsilon^{pqr} I^r, \qquad I^p I^q = J^p J^q  = \ -\delta^{pq} +  \epsilon^{pqr} J^r, \\ [8pt]
&&
\Delta I^p = I^p \Delta = J^p,  \qquad \Delta J^p = J^p \Delta = I^p,
\qquad \Delta^2 = 1 \, .
\end{array}
\ee
Note that the matrices $J^p$ are also the complex structures --- they satisfy all the conditions of {\bf Definition 1}. In constrast to $I^p$, they in addition satisfy the quaternion algebra \p{quatern}.

Consider the matrices $I^p_{\pm} = \frac 12 (I^p \pm J^p) $ and $\Delta_\pm = \frac 12 (1 \pm \Delta)$.
It is clearly seen that the two subalgebras generated by $\{\Delta_+, I^p_+ \}$ and $\{\Delta_-, I^p_- \}$ are closed,
each of them being isomorphic to the quaternion algebra, and the products of the generators from two different
sets  vanish:
\be
I^p_\pm I^q_\pm =  -\delta^{pq} \Delta_\pm +  \epsilon^{pqr} I^r_\pm, \qquad I^p_\pm \Delta_\pm =
\Delta_\pm   I^p_\pm = I^p_\pm, \qquad
\Delta_\pm I^p_\mp = I^p_\mp \Delta_\pm = I^p_\pm I^q_\mp = 0 \, .
\ee

  \end{proof}

This simple mathematical fact gives a primary justification for calling the sigma models, based on the Clifford algebra \p{Clifford},  bi-HKT models. Further justifications will follow.

It makes sense to emphasize here that, according to {\bf Definitions 4, 6}, bi-HKT geometries and
HKT geometries are different. If writing instead in {\bf Definition 6} ``... not {\it necessarily} the quaternion
algebra...,'' the HKT geometry would be a particular case of the more general bi-HKT geometry. But
we prefer to call a manifold bi-HKT if both sectors ${\mathcal H}_+$ and ${\mathcal H}_-$ are not empty.

 \vspace{1mm}

{\bf Corollary}. Bearing in mind {\bf Proposition 1}, it follows that a basis can localy
 be chosen where
\be
\label{IandJ}
\begin{array}{cll}
&I^1_{bi-HKT} \ = \ {\rm diag} (\underbrace{\EuScript{I},\ldots, \EuScript{I}}_{n^*},\ \underbrace{-\EuScript{I}, \ldots, -\EuScript{I}}_{m^*}), \ \ \ \ \ \ \ \
&J^1_{bi-HKT} \ = \ {\rm diag} (\underbrace{\EuScript{I},\ldots,
\EuScript{I}}_{n^*+m^*}),  \\ [7pt]
&I^2_{bi-HKT} \ = \ {\rm diag} (\underbrace{\EuScript{J},\ldots, \EuScript{J}}_{n^*}, \  \underbrace{-\EuScript{J}, \ldots, -\EuScript{J}}_{m^*}), \ \ \ \ \ \ \ \
&J^2_{bi-HKT} \ = \ {\rm diag} (\underbrace{\EuScript{J},\ldots,
\EuScript{J}}_{n^*+m^*}),  \\ [7pt]
&I^3_{bi-HKT} \ = \ {\rm diag} (\underbrace{\EuScript{K},\ldots, \EuScript{K}}_{n^*}, \underbrace{ -\EuScript{K} , \ldots, -\EuScript{K}}_{m^*}),  \ \ \ \ \ \ \ \
& J^3_{bi-HKT} \ = \ {\rm diag} (\underbrace{\EuScript{K},\ldots,
\EuScript{K}}_{n^*+m^*}) \, .
\end{array}
\ee
That means in particular that the  Nijenhuis concomitants,
    \be
\frac 12\, N^\lambda_{\mu\nu} (p,q) \ =\ \Big\{ (I^p)_{[\mu}{}^\sigma \partial_\sigma  (I^q)_{\nu ]}{}^\lambda -
\partial_{[\mu}  (I^q)_{\nu ]}{}^\sigma (I^p)_\sigma{}^\lambda \Big\} + (p \leftrightarrow q) \, ,
    \ee
vanish for any pair of complex structures $\{I^p, I^q\}$ (and also  $\{J^p, J^q\}, \ \{I^p, J^q\}$).

{\bf Remark 1}. As was noted, in contrast to the complex structures $I^p$, the complex structures $J^p$ are quaternionic. The existence of a triple of quaternionic complex structures means that a bi-HKT manifold belongs to the class of {\it hypercomplex} manifolds as defined in \cite{Boyer}.  Bearing this in mind, one can {\it conjecture}  that any hypercomplex manifold is either HKT or bi-HKT.

We  now introduce besides the ordinary linear  ({\bf 4},\,{\bf 4},\,{\bf 0}) multiplet ${\cal V}^{i\alpha}$ , a {\it mirror} pseudoreal
multiplet ${\cal W}^{i' \alpha'}$ satisfying the constraints
   \begin{equation}\label{const2}
D^{i ({j}^\prime}{\cal W}^{k')\alpha'} = 0\, .
     \end{equation}
The mirror multiplet carries the index $i'$ of the right automorphic group $SU_R(2)$ and  a new global Pauli-G\"ursey  index $\alpha'$  having nothing to do with $\alpha$.

We choose the complex superfields ${\cal W}^\mu$ as follows,
  \be
\lb{W-compl}
{\cal W}^{i'\alpha'} \ =\  \left( \begin{array}{cc}  \bar {\cal W}^1 &   \bar{\cal W}^2 \\  {\cal W}^2
&  - {\cal W}^1
\end{array} \right)
 \ee
  Then the constraints \p{const2} acquire the form
\be
\label{cons_thetet_W}
\begin{array}{c}
\bar D_\theta {\cal W}^\mu =  D_\eta {\cal W}^\mu = 0 \,, \qquad
 D_\theta \bar {\cal W}^{\bar \mu} =  \bar
D_\eta \bar {\cal W}^{\bar \mu} = 0 \,,
\\ [7pt]
D_\theta {\cal W}^{\mu}-\epsilon^{\mu\nu}D_\eta\bar{\cal W}^{\bar{\nu}}=0 \,,\qquad
\bar D_\eta {\cal W}^{\mu}+\epsilon^{\mu\nu}\bar D_\theta\bar{\cal W}^{\bar{\nu}}=0  \, .
\end{array}
\ee
Their solution reads
 \be
\lb{W-mu-expan}
{\cal W}^{\mu}=W^{\mu} + \bar\eta\, \epsilon^{\mu\nu}\bar D\bar W^{\bar {\nu}} +i\eta\bar\eta\, \dot{W}^{\mu} \, .
 \ee
with
  \be
\lb{chir-comp}
W^\mu = w^\mu + \sqrt{2} \theta \chi^\mu - i \theta \bar \theta \dot{w}^\mu \, .
   \ee

${\cal W}^\mu$ has exactly the same form as ${\cal V}^m$ up to the interchange $\eta
\leftrightarrow \bar \eta$. This explains our convention \p{W-compl} with the bars in the first line and not
in the second line as in \p{V-compl}. We have chosen it for the  ${\cal N} = 4$ superfields ${\cal W}^\mu$ to involve usual chiral rather than antichiral
${\cal N} = 2$ superfields $W^\mu$. Note that, if introducing the real 4-vector fields ${\cal W}^M$ as in
  \p{Vvect}, their relationship with  ${\cal W}^\mu$ is the same as in \p{v12} with $i \to -i$,
  \be
\lb{w12}
 {\cal W}^{\mu=1} \ =\ {\cal W}^{M=3} +  i{\cal W}^{M=4}, \ \ \ \ \ \ \ \ \ \ {\cal W}^{\mu=2} \ =\
{\cal W}^{M=1} + i{\cal W}^{M=2} \, .
 \ee

It is clear that the Lagrangian like \p{act1},
but with the superpotential depending on $n^*$ mirror multiplets
${\cal W}_a$  has exactly the same component expression.
It is an HKT SQM model.
Bi-HKT models follow from the superpotential ${\cal L}$
depending on {\it both}  ordinary and mirror multiplets. To obtain the model  with the complex structures \p{IandJ}, we have
to take $n^*$ ordinary and $m^*$ mirror multiplets.

In \cite{FIS2}, we gave detailed formulas for the Lagrangian and the supercharges in the simplest nontrivial case
with one ordinary and one mirror multiplets. In what follows, we repeat this analysis for generic nonvanishing $n^*,m^*$.
 This will allow us to write simple generic expression for the supercharges.

\subsection{Lagrangian}

Consider the superfield action
\be
\lb{act4}
S={ \frac14}\displaystyle{\int}dt\, d\theta d\bar\theta\, d\eta d\bar\eta\,
{\cal L}({\cal V}_a^n,\bar{\cal V}_a^{\bar n},
{\cal W}_\alpha^\mu,\bar{\cal W}_\alpha^{\bar \mu})\, ,
\ee
$a = 1,\ldots, n^*; \ \alpha = 1, \ldots, m^*$;  $n,\mu = 1,2$.
(Hopefully, the reader will not confuse the flavor index $\alpha$ just introduced with the Pauli-G\"ursey index $\alpha$ in \p{const1}. Sorry, but the Latin and also the Greek alphabets have finite length.)

Integrating it over  $d\eta d\bar\eta$, we may express the action in ${\cal N} =2$ superspace as
  \be
\lb{act4-2}
S={ \frac14}\displaystyle{\int}dt\, d\theta d\bar\theta\, {\cal L}^\prime(V, \bar V, W, \bar W) \, ,
  \ee
with
\be
\lb{lagr4-2}
\begin{array}{rcl}
{\cal L}^\prime&=&( \Delta^{ab}_{m\bar n}{\cal L})\,DV^{m}_{a}\bar D\bar V^{\bar n}_{b}
- ( \Delta^{\alpha\beta}_{\mu\bar \nu}{\cal L})\,DW^{\mu}_{\alpha}\bar D\bar W^{\bar \nu}_{\beta} \\  [11pt]
&&
- \,\epsilon_{mn}\epsilon_{\mu\nu}( \bar\partial^{a}_{\bar n}\bar\partial^{\alpha}_{\bar \nu}
{\cal L})\,DV^{m}_{a}DW^{\mu}_{\alpha}
+ \epsilon_{\bar m\bar n}\epsilon_{\bar \mu\bar \nu}( \partial^{a}_{n}
\partial^{\alpha}_{\nu}{\cal L})\,\bar D\bar V^{\bar m}_{a}\bar D\bar W^{\bar \mu}_{\alpha}\,.
\end{array}
\ee
Here $\partial^{\,a}_m= {\partial}/{\partial { V}_a^m}$,
$\bar\partial^{\,a}_{\bar m}= {\partial}/{\partial \bar { V}_a^{\bar m}}$,
$\partial^{\,\alpha}_\mu= {\partial}/{\partial { W}_\alpha^\mu}$,
$\bar\partial^{\,\alpha}_{\bar \mu}= {\partial}/{\partial \bar { W}_\alpha^{\bar \mu}}$ and the structures
     \begin{equation} \lb{Del2a-VW}
\Delta_{m\bar n}^{ab}{\cal L} \ =\  \partial^a_m \bar \partial_{\bar n}^b {\cal L}   +
\epsilon_{m k} \epsilon_{\bar n \bar l} \bar \partial^a_{\bar k}  \partial^b_l  {\cal L}\,, \ \ \ \ \ \ \
   \Delta_{\mu\bar \nu}^{\alpha\beta}{\cal L}  \ =\ \partial^\alpha_\mu \partial^\beta_{\bar \nu}
 {\cal L}  +
\epsilon_{\mu\lambda} \epsilon_{\bar\nu \bar\rho}  \bar \partial^\alpha_{\bar \lambda}  \partial^\beta_\rho
{\cal L}
   \end{equation}
depend on $2n^* + 2 m^*$ chiral ${\cal N} = 2$
superfields $V^m_a,\bar V^{\bar m}_a, W^\mu_\alpha, \bar W^{\bar \mu}_\alpha$.
The covariant derivatives $D \equiv D_\theta$ and  $\bar D \equiv \bar D_\theta $ were defined in \p{D_thetet}.

 The Lagrangian \p{lagr4-2} involves the terms $\sim D \bar D$ and also the terms $\sim DD$ and $\sim\bar D \bar D$.
It belongs to the class of Dolbeault SQM models twisted by holomorphic torsions \cite{Hull,FIS1}. A half of the
supersymmetries of the original action \p{act4} is realized manifestly in \p{lagr4-2}.
Another half is ``hidden'' in this formulation. It is implemented as the invariance with respect to the following transformations of ${\cal N} = 2$
superfields,
\be
\lb{SUSYtran-N=2hid-VW}
\begin{array}{cll}
&\delta_\eta V_a^{m} = -\varepsilon_\eta\, \epsilon^{mn}\bar D\bar V_a^{\bar n}\,, \qquad
&\delta_\eta \bar V_a^{\bar{m}} = \bar\varepsilon_\eta\, \epsilon^{{\bar m}\bar{n}} D V_a^{n}\,,
\\ [7pt]
&\delta_\eta W_\alpha^{\mu} = -\bar\varepsilon_\eta\, \epsilon^{\mu\nu}\bar D\bar W_\alpha^{\bar {\nu}}\,, \qquad
&\delta_\eta \bar W_\alpha^{\bar{\mu}} = \varepsilon_\eta\, \epsilon^{{\bar \mu}\bar{\nu}} D W_\alpha^{\nu}\,.
\end{array}
\ee

With the Lagrangian \p{lagr4-2} at hand, we can go down to components. The second derivatives \p{Del2a-VW} of the prepotential give the metric; the bosonic part of the  component Lagrangian reads
\be
\lb{lagr4-b-app}
L_{b} =
h^{ab}_{m\bar n}\, \dot v_a^m \dot {\bar v}_b^{\bar n} + h^{\alpha\beta}_{\mu\bar \nu}\,
\dot w_\alpha^\mu \dot {\bar w}_\beta^{\bar \nu} \, ,
\ee
with
\begin{equation} \lb{metr}
h^{ab}_{m\bar n}=\Delta^{ab}_{m\bar n} {\cal L} \,,
\qquad h^{\alpha\beta}_{\mu\bar \nu}=-\Delta^{\alpha\beta}_{\mu\bar \nu} {\cal L}\,,
\end{equation}
and the derivatives are taken now with respect to the bosonic fields $v^m_a, w_\alpha^\mu$.

One can notice that
\begin{equation} \lb{prop-metr}
\begin{array}{c}
\epsilon_{nk}h^{ab}_{m\bar k}=-\epsilon_{mk}h^{ba}_{n\bar k}\,,\qquad
\epsilon_{\bar n\bar k}h^{ab}_{k\bar m}=-\epsilon_{\bar m\bar k}h^{ba}_{k\bar n}\,, \\ [8pt]
\epsilon_{\nu\lambda}h^{\alpha\beta}_{\mu\bar\lambda}=-\epsilon_{\mu\lambda}h^{\beta\alpha}_{\nu\bar\lambda}\,,\qquad
\epsilon_{\bar\nu\bar\lambda}h^{\alpha\beta}_{\lambda\bar\mu} = -\epsilon_{\bar\mu\bar\lambda}h^{\beta\alpha}_{\lambda\bar\nu}\,.
\end{array}
\end{equation}

 The  Lagrangian involves also the 2-fermion and 4-fermion terms depending on the variables
 $\psi^n_a$  (the fermion superpartners of $v_a^n$) and $\chi^\mu_\alpha$  (the fermion superpartners of $w_\alpha^\mu$).
They are spelled out in the
Appendix.

\subsection{Supercharges}

The supercharges can be found by the standard N\"other method.
A pair of complex conjugated supercharges that correspond to the manifest ${\cal N} = 2$ supersymmetry can be taken from Ref.\cite{FIS1}. They can be represented in the following form,
  \begin{equation}\label{S-CKT}
\begin{array}{rcl}
S&=&\sqrt{2}\,\psi_a^m\left(\Pi^a_m-{ \frac{i}{2}}\,\partial^a_{m}h^{bc}_{k\bar l}\, \psi_b^k\bar\psi_{c}^{\bar l}
-{ \frac{i}{2}}\,\partial^a_{m} h^{\beta\gamma}_{\nu\bar \lambda}\,
\chi_\beta^\nu\bar\chi_{\gamma}^{\bar \lambda} +
{ \frac{i}{2}}\,\epsilon_{mn}  \epsilon_{\nu\lambda}   \, \bar\partial^a_{\bar n} h^{\alpha\beta}_{\mu\bar \lambda}\,
\chi_\alpha^\mu\chi_{\beta}^{\nu}\right)
\\[9pt]
&&+\sqrt{2}\,\chi_\alpha^\mu\left(\Pi^\alpha_\mu-{ \frac{i}{2}}\,\partial^\alpha_{\mu}h^{\beta\gamma}_{\nu\bar \lambda}\,
\chi_\beta^\nu\bar\chi_{\gamma}^{\bar \lambda}-{ \frac{i}{2}}\,\partial^\alpha_{\mu}h^{bc}_{k\bar l}\, \psi_b^k\bar\psi_{c}^{\bar l}
+ { \frac{i}{2}}\,\epsilon_{\mu\nu}
\epsilon_{nk} \, \bar\partial^\alpha_{\bar\nu} h^{ab}_{m\bar k}\, \psi_a^m\psi_{b}^{n}\right),\\ [12pt]
    \bar S&=&\sqrt{2}\,\bar\psi_a^{\bar m}\left(\bar\Pi^a_{\bar m}+{ \frac{i}{2}}\,\bar\partial^a_{\bar m}h^{bc}_{k\bar l}\,
\psi_b^k\bar\psi_{c}^{\bar l}
+{ \frac{i}{2}}\,\bar\partial^a_{\bar m}h^{\beta\gamma}_{\nu\bar \lambda}\,
\chi_\beta^\nu\bar\chi_{\gamma}^{\bar \lambda}
- { \frac{i}{2}}\,\epsilon_{\bar m\bar n}\epsilon_{\bar\mu\bar\lambda}  \,  \partial^a_{n} h^{\alpha\beta}_{\lambda\bar \nu}\,
\bar\chi_\alpha^{\bar\mu}\bar\chi_{\beta}^{\bar\nu}\right)\\[9pt]
&&+\sqrt{2}\,\bar\chi_\alpha^{\bar\mu}\left(\bar\Pi^\alpha_{\bar\mu}+{ \frac{i}{2}}\,
\bar\partial^\alpha_{\bar\mu}h^{\beta\gamma}_{\nu\bar \lambda}\,
\chi_\beta^\nu\bar\chi_{\gamma}^{\bar \lambda}+{ \frac{i}{2}}\,\bar\partial^\alpha_{\bar\mu}h^{bc}_{k\bar l}\,
\psi_b^k\bar\psi_{c}^{\bar l}
- { \frac{i}{2}}\,\epsilon_{\bar\mu\bar\nu}
\epsilon_{\bar m\bar k}  \, \partial^\alpha_{\nu} h^{ab}_{k\bar n}\, \bar\psi_a^{\bar m}\bar\psi_{b}^{\bar n}\right) \,,
\end{array}
\end{equation}
where  the momenta
$\Pi^a_{m}$, $\Pi^\alpha_{\mu}$ are given by the variations of the Lagrangian with respect to  $\dot{v}_a^n$ and
 $\dot{w}_\alpha^\mu$, calculated  while
keeping $\psi^n_a, \chi^\mu_\alpha$ fixed. The fermion charge $F$ is not conserved here due to the presence of
holomorphic terms $\sim DD$ and $\sim \bar D \bar D$ in \p{lagr4-2}. The supercharge S involves the usual terms with $F=1$ and also the holomorphic terms $\sim \psi \chi \chi$ and $\sim \chi \psi \psi$ with $F=3$. Likewise, $\bar S$ involves the terms with $F= -1$ and $F = -3$.

The expressions for the supercharges can be rendered more compact if introducing the  notation
\begin{equation} \lb{ferm-B-c}
\Psi^{\cal M} =\left\{\psi_a^m,\chi_\alpha^\mu\right\},  \quad \bar\Psi^{\bar {\cal M}} =\left\{\bar\psi_a^{\bar m},\bar\chi_\alpha^{\bar \mu}\right\}\, ,
\qquad
\Pi_{\cal M} =\left\{\Pi^a_m,\Pi^\alpha_\mu\right\},\quad \partial_{\cal M} =\left\{\partial^a_m,\partial^\alpha_\mu\right\}.
\end{equation}
   In addition, we introduce the metric tensor
\begin{equation} \lb{metr-tens-B-c}
h_{{\cal M}\bar {\cal N}}=\left(
\begin{array}{cc}
h^{ab}_{m\bar n} & 0 \\
0 & h^{\alpha\beta}_{\mu\bar\nu} \\
\end{array}
\right),
\end{equation}
and {\it two} hypercomplex  structure matrices corresponding to two triples of complex structures in \p{IandJ}
\begin{equation} \lb{comp-str-B-c}
{\cal I}_{{\cal M}}{}^{\bar {\cal N}}= \bar {\cal I}_{\bar {\cal M}}{}^{\cal N}=\left(
\begin{array}{cc}
\epsilon_{mn}\delta^{a}{}_{b} & 0 \\
0 & -\epsilon_{\mu\nu}\delta^{\alpha}{}_{\beta} \\
\end{array}
\right),
\qquad
{\cal J}_{\cal M}{}^{\bar {\cal N}}= \bar {\cal J}_{\bar {\cal M}}{}^{\cal N}=\left(
\begin{array}{cc}
\epsilon_{mn}\delta^{a}{}_{b} & 0 \\
0 & \epsilon_{\mu\nu}\delta^{\alpha}{}_{\beta} \\
\end{array}
\right).
  \end{equation}
The supercharges \p{S-CKT} acquire   the form
\begin{equation}\label{S-B-c}
\begin{array}{rcl}
S&=&\sqrt{2}\,\Psi^{\cal M} \Big(\Pi_{\cal M} \,-\, { \frac{i}{2}}\,\partial_{\cal M}
 h_{{\cal K}\bar {\cal L}}\Psi^{\cal K}\bar\Psi^{\cal L}
\,+\, { \frac{i}{2}}\,\partial_{\cal M} {\cal C}_{{\cal K}{\cal L}}\, \Psi^{\cal K} \Psi^{\cal L}\Big)
\,,
\\ [14pt]
\bar S&=&\sqrt{2}\,\bar\Psi^{\bar {\cal M}} \Big(\bar\Pi_{\bar {\cal M}} \,+\,
 { \frac{i}{2}}\,\bar\partial_{\bar {\cal M}} h_{{\cal K}\bar {\cal L}}\Psi^{\cal K}\bar\Psi^{\cal L}
\,-\, { \frac{i}{2}}\,\bar\partial_{\bar {\cal M}}
\bar{\cal C}_{\bar {\cal K}\bar {\cal L}}\, \bar\Psi^{\bar {\cal K}}
\bar\Psi^{\bar {\cal L}}\Big)
\,,
\end{array}
   \end{equation}
where
   \begin{equation}\label{B-calB-c}
{\cal C}_{{\cal K}{\cal L}}= -{\cal I}_{[{\cal K}}{}^{\bar{\cal P}} {\cal J}_{{\cal L}]}^
{\bar {\cal R}}\,\partial_{\bar {\cal P}}\partial_{\bar {\cal R}}\,{\cal L}\,.
   \end{equation}

To find the second pair of the component supercharges corresponding to the transformations \p{SUSYtran-N=2hid-VW}, we note that the terms \p{lagr4-b-app}, \p{lagr4-2f-app}, and \p{lagr4-4f-app} in the
component Lagrangian are all invariant with respect to
the following $Z_2 \times Z_2$  transformations,
  \begin{equation} \label{dis-tr}
\begin{array}{c}
\chi^{\mu}_{\alpha}\rightarrow \epsilon^{\mu \nu}\bar\chi^{\bar \nu}_{\alpha}\,,\qquad
\bar\chi^{\bar \mu}_{\alpha}\rightarrow \epsilon^{\bar \mu\bar \nu}\chi^{\nu}_{\alpha}\,; \\ [7pt]
v^{m}_{a}\rightarrow \epsilon^{mn}\bar v^{\bar n}_{a}\,,\qquad
\bar v^{\bar m}_{a}\rightarrow \epsilon^{\bar m \bar n} v^{n}_{a}\,;\qquad
\partial_{m}^{a}\rightarrow - \epsilon_{mn}\bar \partial_{\bar n}^{a}\,,\qquad
\bar \partial_{\bar m}^{a}\rightarrow - \epsilon_{\bar m \bar n} \partial_{n}^{a} \, ,
\end{array}
\end{equation}
(the fields $w^{\mu}_{\alpha}$ and $\psi^{m}_{a}$ are not transformed), and
  \begin{equation} \label{dis-tr-1}
\begin{array}{c}
\psi^{m}_{a}\rightarrow \epsilon^{mn}\bar\psi^{\bar n}_{a}\,,\qquad
\bar\psi^{\bar m}_{a}\rightarrow \epsilon^{\bar m\bar n}\psi^{n}_{a}\,; \\ [7pt]
w^{\mu}_{\alpha}\rightarrow \epsilon^{\mu\nu}\bar w^{\bar \nu}_{\alpha}\,,\qquad
\bar w^{\bar \mu}_{\alpha}\rightarrow \epsilon^{\bar \mu\bar \nu} w^{\nu}_{\alpha}\,;\qquad
\partial_{\mu}^{\alpha}\rightarrow - \epsilon_{\mu\nu}\bar \partial_{\bar \nu}^{\alpha}\,,\qquad
\bar \partial_{\bar \mu}^{\alpha}\rightarrow - \epsilon_{\bar \mu\bar \nu} \partial_{\nu}^{\alpha}
\end{array}
\end{equation}
(the fields $v^{m}_{a}$ and $\chi^{\mu}_{\alpha}$ are not transformed).

Such discrete symmetries is a common feature of extended supersymmetric models.
They appear as discrete subgroups of continuous $R$-symmetries related to nontrivial automorphisms of supersymmetry algebra. In our case, the symmetry \p{dis-tr} is related to the automorphic $SU_L(2)$ symmetry that rotates the index $i$ of $\theta^{ik'}$. As is clear from the component expansions of the ordinary and the mirror multiplets (see Eqs. (2.3) and (2.41) in Ref. \cite{FIS2} ), the latter symmetry transforms, indeed, the fields $v^m$ and $\chi^\mu$   leaving  $w^\mu$  and  $\psi^m$  intact.
 Analogously, the symmetry \p{dis-tr-1} is related to the automorphic $SU_R(2)$  that
 acts on  the index $k'$ of  $\theta^{ik'}$. It rotates $w^\mu$  and  $\psi^m$   leaving  $v^m$ and $\chi^\mu$ intact.

The second pair of the supercharges can be obtained by acting on \p{S-CKT}
with the discrete transformations \p{dis-tr} supplemented by
$\Pi_{m}^{a}\rightarrow - \epsilon_{mn}\bar \Pi_{\bar n}^{a}$ or
with \p{dis-tr-1}
supplemented by $\Pi_{\mu}^{\alpha}\rightarrow - \epsilon_{\mu\nu}\bar \Pi_{\bar \nu}^{\alpha}$.
(In the first case, $S$  goes over to  $R$ and  $\bar S$  to  $\bar R$, and in the second,  $S \to \bar R, \bar S \to R$).

We obtain
\begin{equation}\label{R-CKT}
\begin{array}{rcl}
R&=&\sqrt{2}\,
\psi_a^{n}\,\epsilon_{nm}\left(\bar\Pi^a_{\bar m} - { \frac{i}{2}}\,\bar\partial^a_{\bar m}h^{b c}_{k\bar l}\, \psi_b^k\bar\psi_{ c}^{\bar l}
+ { \frac{i}{2}}\,\bar\partial^a_{\bar m}h^{\beta\gamma}_{\nu\bar \lambda}\,
\chi_\beta^\nu\bar\chi_{\gamma}^{\bar \lambda}
- { \frac{i}{2}}\,\epsilon_{\bar m\bar k} \epsilon_{\bar\mu\bar\lambda}  \, \partial^a_{k}h^{\alpha\beta}_{\lambda\bar \nu}\,
\bar\chi_\alpha^{\bar\mu}\bar\chi_{\beta}^{\bar\nu}\right)\\[9pt]
&&- \sqrt{2}\,
\bar\chi_\alpha^{\bar\nu}\,\epsilon_{\bar\nu\bar\mu}\left(\Pi^\alpha_\mu  +
{ \frac{i}{2}}\,\partial^\alpha_{\mu}h^{\beta\gamma}_{\nu\bar \lambda}\,
\chi_\beta^\nu\bar\chi_{\gamma}^{\bar \lambda}-{ \frac{i}{2}}\,\partial^\alpha_{\mu}h^{bc}_{k\bar l}\, \psi_b^k\bar\psi_{c}^{\bar l}
 + { \frac{i}{2}}\,\epsilon_{\mu\lambda}
\epsilon_{nk}  \, \bar\partial^\alpha_{\bar\lambda} h^{ab}_{m\bar k}\, \psi_a^m\psi_{b}^{n}\right)
\,,\\ [12pt]
\bar R&=&\sqrt{2}\,
\bar\psi_a^{\bar n}\,\epsilon_{\bar n\bar m}\left(\Pi^a_m+{ \frac{i}{2}}\,
\partial^a_{m}h^{bc}_{k\bar l}\, \psi_b^k\bar\psi_{c}^{\bar l}
-{  \frac{i}{2}}\,\partial^a_{m}h^{\beta\gamma}_{\nu\bar \lambda}\,
\chi_\beta^\nu\bar\chi_{\gamma}^{\bar \lambda}
+ { \frac{i}{2}}\,\epsilon_{mk} \epsilon_{\nu\lambda} \, \bar\partial^a_{\bar k}h^{\alpha\beta}_{\mu\bar \lambda}\,
\chi_\alpha^\mu\chi_{\beta}^{\nu}\right)
\\[9pt]
&&- \sqrt{2}\,
\chi_\alpha^\nu\,\epsilon_{\nu\mu}\left(\bar\Pi^\alpha_{\bar\mu}-{ \frac{i}{2}}\,
\bar\partial^\alpha_{\bar\mu}h^{\beta\gamma}_{\nu\bar \lambda}\,
\chi_\beta^\nu\bar\chi_{\gamma}^{\bar \lambda}+{ \frac{i}{2}}\,\bar\partial^\alpha_{\bar\mu}h^{bc}_{k\bar l}\,
\psi_b^k\bar\psi_{c}^{\bar l}
- { \frac{i}{2}}\,\epsilon_{\bar\mu\bar\lambda}
\epsilon_{\bar m\bar k}  \, \partial^\alpha_{\lambda} h^{ab}_{k\bar n}\, \bar\psi_a^{\bar m}\bar\psi_{b}^{\bar n}\right) \, .
\end{array}
\end{equation}

If suppressing in \p{S-CKT} and \p{R-CKT} the terms involving $\chi^\mu_\alpha$ {\it or} the terms involving $\psi^n_a$, we are 
would reproduce the expressions \p{SR-HKT-s} for the HKT supercharges. There are, however, also nontrivial mixed terms.

    Similar to $S$ and  $\bar S$,  one can derive more  compact transparent expressions for   $R, \bar R$, if introducing
the twisted fermion, momenta and derivative multiplets obtained by the action of the $Z_2$ symmetry
\p{dis-tr}  on the multiplets \p{ferm-B-c},
\begin{equation} \lb{ferm-B-c-Phi}
\Psi^{* \cal M} = \left\{\epsilon^{mn} \bar \psi_a^{\bar n},  \chi^\mu_\alpha \right\}, \ \ \ \ \ \
\partial^*_{\cal M}  =   \left\{ \partial_m^a, -\epsilon_{\mu\nu} \bar \partial_{\bar \nu}^\alpha \right \}, \ \ \ \ \ \
\Pi^*_{\cal M}  =  \left\{ \Pi_m^a, -\epsilon_{\mu\nu} \bar \Pi_{\bar \nu}^\alpha \right\},
\end{equation}

Then $R, \bar R$ are expressed into \p{ferm-B-c-Phi} exactly in the same way as $S, \bar S$ are expressed into
\p{ferm-B-c}.
When written in such a form, both pairs of the supercharges depend on two hypercomplex structures ${\cal I}$ and
${\cal J}$ in the same symmetric way.

Note the following. As was shown in \cite{taming}, in many cases a set of quantum complex supercharges
in a nontrivial SQM model, can be obtained from the set of complex supercharges of a free model by two operations:
 {\it i)} a similarity transformation (that needs not to be a unitary transformation) and
{\it ii)} Hamiltonian reduction. This applies in particular to the quantum counterparts of the supercharges \p{SR-HKT-s} of the
HKT model, one can represent them as
\be
\label{HKT-simil}
\begin{array}{lcl}
 \hat{S}^{\rm HKT} & =& \sqrt{2} \exp \left\{ \omega_{B   C} \psi^B \bar \psi^{ C} \right\}
\psi^A P_A \exp \left\{- \omega_{B   C} \psi^B \bar \psi^{ C} \right\} ,
\\ [8pt]
\hat{R}^{\rm HKT} & =& \sqrt{2} \exp \left\{\omega_{B  C} \psi^B \bar \psi^{C} \right\}
\psi^A {\cal I}_A{}^{D} \bar P_{ D} \exp \left\{- \omega_{B  C} \psi^B \bar \psi^{ C} \right\}
\, ,
\end{array}
 \ee
where $\psi^B$ and $\bar\psi^{ C}$ are the tangent space canonic fermion variables,
$\{\psi^B, \bar \psi^{ B} \}_{\rm P.B.} = -i \delta^{AB}$,  $\hat{P}_A = -i \partial_A$ are the flat canonical momenta that commute
with  $\psi^B$ and $\bar \psi^{ C}$; $\omega_{B  C} $ is an arbitrary complex matrix.

Let us explain it in some more details (in \cite{taming}, this statement was made only for the simplest case of a 4-dimensional conformally flat manifold).
Let $U = \exp \left\{ \omega_{B  C} \psi^B \bar \psi^{ C} \right\}  $. Using the Hadamard formula, one can show that
 \be
\lb{psi-sim}
\begin{array}{lcl}
U\psi^C U^{-1}  &=&  \psi^A  (e^\omega)_{AC} \,, \\ [7pt]
 U \partial_C U^{-1} &=& \partial_C +  (e^\omega)_{AD}  (\partial_C e^\omega)_{DB} \psi^A \bar \psi^{ B} \, .
\end{array}
  \ee

The matrices $e^{\pm \omega}$,  $e^{\pm \omega^\dagger}$ can be interpreted as the complex vielbeins,
\be
\lb{vielbein}
(e^\omega)_{AC} \ \to \ e^{\cal J}_A, \quad (e^{-\omega})_{CA} \ \to \ e^A_{\cal J} , \quad (e^{\omega^\dagger})_{CA} \ \to \
\bar e^{\bar {\cal J}}_{ A}, \quad       (e^{-\omega^\dagger})_{AC} \ \to \
\bar e_{\bar {\cal J}}^{A}\, .
 \ee
We obtain,
\be
\lb{S-cherez-P}
\hat{S}^{\rm HKT} \ =\ \sqrt{2} \, \psi^{\cal M} \left[ \hat{P}_{\cal M} - i ({\bar e}^{B}_{\bar {\cal K}}\,
\partial_{\cal M} e^B_{\cal J} ) \psi^{\cal J} \bar\psi^{\bar{\cal K}} \right] \, .
\ee
The same expression describes the classical complex supercharge, with the canonical momenta $P_{\cal M}$  being
the variation of the Lagrangian with respect to $\dot{z}^{\cal M}$ with fixed $\psi^A$, $\bar\psi^A$. To make contact with
\p{SR-HKT-s}, one should go over to the momenta $\Pi_{\cal M}$ representing the variations, calculated while keeping fixed the fermions with the world indices.  $P_{\cal M}$ is  related to $\Pi_{\cal M}$ as follows \cite{quasi},
 \be
 \lb{Pi-P}
P_{\cal M} \ =\ \Pi_{\cal M} + \frac i2 \left[\bar e^{ B}_{\bar {\cal K}} \, (\partial_{\cal M} e^B_{\cal J})
-   (\partial_{\cal M}  \bar e^{ B}_{\bar {\cal K}} ) \, e^B_{\cal J}  \right] \psi^{\cal J} \bar\psi^{\bar {\cal K}}   \, .
  \ee
The first line in \p{SR-HKT-s} is thus reproduced.

By the same token, the second line in \p{HKT-simil} coincides with the third line in \p{SR-HKT-s}, the  matrix
$ {\cal I}_A{}^{ D} $  going over to  ${\cal I}_{\cal M}{}^{\bar {\cal N}}$.
  In this case, the supercharges $S$ and $R$ are conjugated by {\it the same} operator.

On the other hand, for a bi-HKT model, such a universal similarity transformation does not exist.
One can transform the flat supercharge $S_0$ to the expression in \p{S-CKT} and the flat supercharge $R_0$ to the
expression in \p{R-CKT}, but the corresponding operators $U$ are {\it different}. For example, the terms $\sim\!\psi\chi\chi$ and
$\sim\!\chi\psi\psi$ in $ S$ written in \p{S-CKT} are obtained from $S_0$ by a similarity transformation with the operator
$\exp\{ \sim\!\psi \chi \}$, whereas the corresponding terms $\sim\!\psi \bar\chi \bar\chi$ and $\sim\!\bar\chi \psi\psi$ in
$R$ written in \p{R-CKT} are obtained from $R_0$ by a similarity transformation with the operator
$\exp\{ \sim\!\psi \bar\chi \}$.

\subsection{Real supercharges and geometry.}

We will prove here the following theorem
(it represents a generalization of a similar theorem proved 
in \cite{FIS2} for  the simplest bi-HKT model with one ordinary and one mirror multiplet):
\begin{thm}
In the system, described by the superfield Lagrangian \p{act4},  the supercharges \p{S-CKT} and \p{R-CKT} are equivalent to the following four real supercharges,
 \be
\lb{Q_real}
\begin{array}{lcl}
Q &=&  \Psi^M \left(P_M - {\displaystyle \frac i2}\, \Omega_{MNL}  \Psi^N  \Psi^L +
{\displaystyle \frac i{12}}\, C_{MNL}   \Psi^N  \Psi^L \right)  \, , \\ [9pt]
Q^p & =&  \Psi^S (I^p)_S^{\ M}  \left( P_M  - {\displaystyle \frac i2}\,  \Omega_{MNL}  \Psi^N  \Psi^L
- {\displaystyle \frac i4}\,  B^p_{MNL}  \Psi^N  \Psi^L + {\displaystyle \frac i{12}}\,
H^p_{MNL}  \Psi^N  \Psi^L \right) \, ,
\end{array}
  \ee
where $\Psi^M = \{ \Psi^{\cal M}, \bar\Psi^{\bar{\cal M}} \}$; the canonical momenta $P_M =
\{P_{\cal M}, \bar P_{\bar {\cal M}}\}$ are calculated, while keeping $\Psi^A$ fixed, and are related
to $\Pi_M$ as in \p{Pi-P}; the complex structures $I^p$ are given in \p{IandJ}, \p{IJKcanon};
  \be
\lb{Omega}
\Omega_{MNL} = e^A_N e^B_L  \Omega_{M, AB} \ =\
e^A_N e^B_L  e_{AS}  \left( \partial_M e^S_B + \Gamma^S_{MK} e^K_B \right)
  \ee
are the spin connections; $B^p_{MNL}$ are the Bismut torsions \p{CBismut}
for the complex structure $I^p$; the form associated with the totally antisymmetric $H^p_{MNL}$ is the sum
of the exact holomorphic (3,0) and
antiholomorphic (0,3) forms with respect to $I^p$, such that the sums
$ B^p_{MNL} + H^p_{MNL}$ give for all $p=1,2,3$ one and the same
  full  torsion   tensor $C_{MNL}$ entering the first line in \p{Q_real}.
The supercharges \p{Q_real} satisfy the standard supersymmetry
algebra \p{SUSY-alg}.
\end{thm}

In other words, the Lagrangian \p{act4} describes a ${\cal N} = 4$  sigma model, living
 on a bi-HKT manifold, as defined in the beginning of Sect. 3.

\begin{proof}

\begin{itemize}

\item Consider first the supercharges \p{S-CKT}. Without the holomorphic terms (of fermion charge $F = 3$ in $S$ and of fermion charge $F = -3$ in $\bar S$), the supercharges would have the form
  \be
\lb{S-trunc}
S^{\rm truncated}  &=& \sqrt{2} \Psi^{\cal M} \left[ \Pi_{\cal M} - \frac i2
\partial_{\cal M} h_{{\cal N} \bar {\cal P}} \Psi^{\cal N} \bar\Psi^{\bar {\cal P}} \right] \, ,
   \ee
The expressions \p{S-trunc} coincide with the complex supercharges of the Dolbeault
model (see Eq.(4.2) of Ref.\cite{quasi}). As was shown in \cite{Braden,Mavra,QHKT},
the real and imaginary parts of the supercharges \p{S-trunc} can be presented in the form
\be
\lb{real-Dolb}
\begin{array}{lcl}
Q^{\rm trunc.} &=& {\displaystyle \frac {S^{\rm trunc.} + \bar S^{\rm trunc.}}{\sqrt{2}}}  =
 \Psi^M \left(P_M - {\displaystyle \frac i2}\, \Omega_{MNL}
\Psi^N  \Psi^L + {\displaystyle \frac i{12}}\, B_{MNL}   \Psi^N  \Psi^L \right)  \, , \\ [8pt]
\tilde{Q}^{\rm trunc.} &=& {\displaystyle \frac {i(\bar S^{\rm trunc.} - S^{\rm trunc.})}{\sqrt{2}}} =  \Psi^S I_S^{\ M}  \left( P_M  -
{\displaystyle \frac i2}\, \Omega_{MNL}  \Psi^N  \Psi^L
- {\displaystyle \frac i4}\, B_{MNL}  \Psi^N  \Psi^L  \right) \, ,
\end{array}
\ee
where $I_S^{\ M}$ is a complex structure $I =  -{\rm diag}(\epsilon, \ldots ,
\epsilon)$ and $B_{MNL}$ is its Bismut torsion. The complex structure $I$
coincides with the complex structure $J^1$ in \p{IandJ}.

Consider now the holomorphic terms. They can be presented as
   \be
\lb{S-hol}
\begin{array}{lcl}
S^{\rm hol} &=& {\displaystyle \frac {i\sqrt{2}}{12}}\, \Psi^{\cal M}  \Psi^{\cal N} \Psi^{\cal S}
H^1_{{\cal M} {\cal N} {\cal S}} \\ [8pt]
   \bar S^{\rm hol} &=& {\displaystyle \frac {i\sqrt{2}}{12}}\, \bar \Psi^{\bar {\cal M}} \bar \Psi^{\bar {\cal N}}
   \bar \Psi^{\bar {\cal S}}
H^1_{\bar {\cal M} \bar {\cal N} \bar {\cal S}}\, ,
\end{array}
   \ee
where
 \be
 H^1_{ma, \mu \alpha,  \nu\beta} \ =\  \epsilon_{mn} \epsilon_{\mu\nu} \,\bar \partial^a_{\bar n}
 h_{\mu\bar\lambda}^{\alpha\beta},  \quad \quad
  H^1_{\mu \alpha, ma , nb} \ =\  \epsilon_{\mu\nu} \epsilon_{nk} \,\bar \partial^\alpha_{\bar \nu}
 h_{m \bar k}^{ab} \,,
 \ee
 and the other nonzero components of $H^1_{{\cal M} {\cal N} {\cal S}}$ are restored by antisymmetry.

   The associated (3,0) form is exact, $H^1_{{\cal M} {\cal N} {\cal S}} = \partial_{[{\cal M}}
    {\cal B}_{{\cal N}{\cal S}]}$ with
   \be
   \lb{B-tors}
{\cal B}_{ma, \mu\alpha} = \epsilon_{mn} \epsilon_{\mu\nu} \bar \partial^a_{\bar n}
 \partial^\alpha_{\bar \nu} {\cal L}
  \ee
[this structure was displayed in the superfield Lagrangian  \p{lagr4-2}].

  Adding the real and imaginary parts of $S^{\rm hol}$ to \p{real-Dolb}, we arrive at the supercharge
$Q$  in \p{Q_real} for the real part, while the imaginary part is given by the expression similar to  $Q^1$
in \p{Q_real} , but with $J^1$ standing for $I^1$.

Recall now that the result \p{real-Dolb} was derived under the universal standard convention \p{sign-z} for the expression
of all complex variables via the real ones. However, in our case the conventions in the ordinary and the mirror sectors are different
[see Eqs. \p{v12}, \p{w12}]. In the mirror sector, \p{sign-z} coincides with \p{v12}, while, in the ordinary sector,
the holomorphic and antiholomorphic coordinates are interchanged. Thus, one should change the sign of the components of $I = J^1$ in the
ordinary sector, which gives $-I^1$, and we arrive finally at $-Q^1$ in \p{Q_real}. The overall sign can, of course, be reversed.

\item The superfield action \p{act4} is invariant under the supersymmetry transformations,
  \be
\lb{s-space-tran}
 \theta^{i k'} & \to &  \theta^{i k'}  + \varepsilon^{i k'} \nn
t & \to & t + i \theta^{i k'}  \varepsilon_{ik'} \, .
  \ee
The pseudoreal tensor $\varepsilon^{i k'}$ involves an isosinglet and isotriplet part,
  \be
\lb{iso}
\varepsilon^i_{k'} \ =\ i \varepsilon_0 \delta^i_{k'} + \varepsilon_p
(\sigma^p)^i{}_{k'}
  \ee
with real Grassmann $\varepsilon_0, \varepsilon_p$.
The supercharges $S, \bar S$ are the N\"other charges corresponding to the shift of $\theta$. Bearing in mind
\p{thetet-def} and \p{iso}, this means that the real supercharges $Q,  Q^1$ in Eq.\p{Q_real} are obtained by variating  the Lagrangian with the parameters
$\varepsilon_0, \varepsilon_3$. But there are also the supercharges
$R, \bar R$ corresponding to the shift of $\eta$, which involves in the language of Eq.\p{iso} the parameters
$\varepsilon_1$ and $\varepsilon_2$. The latter represent the real and imaginary parts of the complex parameter $\varepsilon_\eta $ that enters  the transformation
law \p{SUSYtran-N=2hid-VW}. By symmetry, the supercharges associated with the variations $\propto \varepsilon_{2,3}$ should have exactly the same structure
as the supercharge $Q^1$ in \p{Q_real}, but involve  now the complex structures $I^{2,3}$ and coincide (up to a possible extra sign) with $Q^{2,3}$. A precise calculation shows that, while the
supercharge associated with $\varepsilon_3$ involves the complex structure $I^1$, the     supercharge associated with $\varepsilon_1$ involves the complex structure $-I^3$
and the supercharge associated with $\varepsilon_2$ --- the  structure $I^2$. The triple of the Clifford complex structures $(-I^3, I^2, I^1)$ is isomorphic to the triple
$(I^1, I^2, I^3)$ in the full algebra ${\cal H}_+ + {\cal H}_-$. That follows from the fact that the triple $(-\EuScript{K}, \EuScript{J}, \EuScript{I})$ forms the same quaternionic algebra as
$(\EuScript{I}, \EuScript{J}, \EuScript{K})$.

\end{itemize}
\end{proof}

{\bf Remark 2}. The fact that {\it both} $J^1$ and $I^1$ are legitimate complex structures in our model is related to the fact that the metric \p{metr} has a block diagonal
form. Clearly, it is invariant under the interchange $v \leftrightarrow \bar v$ and/or $w \leftrightarrow \bar w$. Thus, using the conventions \p{v12} and \p{w12} is a convenient
option, but not an obligation. One could as well use the universal convention \p{sign-z}. In this case, we would obtain instead of the triple $(-I^3, I^2, I^1)$, the isomorphic
Clifford triple $(-J^3, I^2, J^1)$.

\setcounter{equation}0

\section{Quasicomplex  bi-K\"ahler manifolds.}

Consider a set of $n^*$ superfields representing chiral linear ({\bf 2}, {\bf 4}, {\bf 2}) multiplets,
(it would be interesting also to study the models based on the nonilinear ${\cal N} = 4$ chiral multiplets \cite{IKL,BelKriv}),
\begin{equation}\label{242-sf}
{\cal Z}^a(t; \, \theta,\bar\theta,\eta,\bar\eta)\,,\qquad a=1,\ldots,n^*\,,
\end{equation}
and a set of $m^*$ twisted (mirror) chiral ({\bf 2}, {\bf 4}, {\bf 2}) superfields,
\begin{equation}\label{242-sf-tw}
{\cal U}^\alpha(t; \theta,\bar\theta,\eta,\bar\eta)\,,\qquad \alpha=1,\ldots,m^*\,.
\end{equation}
These superfields satisfy the constraints
\begin{equation}\label{242-const}
\bar D_\theta\,{\cal Z}^a=0\,,\qquad \bar D_\eta\,{\cal Z}^a=0\,,
\end{equation}
\begin{equation}\label{242-const1}
\bar D_\theta\,{\cal U}^\alpha=0\,,\qquad D_\eta\,{\cal U}^\alpha=0\, .
\end{equation}

The superfields \eqref{242-sf}, \eqref{242-sf-tw} can be expressed via  ${\cal N} = 2$  superfields. There are two options.

\subsection{({\bf 2},\,{\bf 4},\,{\bf 2}) = ({\bf 2},\,{\bf 2},\,{\bf 0})$\oplus$({\bf 0},\,{\bf 2},\,{\bf 2})}

One can represent a ({\bf 2},\,{\bf 4},\,{\bf 2}) multiplet as two chiral multiplets of the types   ({\bf 2},\,{\bf 2},\,{\bf 0}) and
({\bf 0},\,{\bf 2},\,{\bf 2})\,  \cite{IS},
\begin{equation}\label{calZa}
{\cal Z}^a  = Z^a + \sqrt{2}\, \eta \, \Phi^a - i \eta \bar \eta\, \dot{Z}^a \, ,\qquad
{\cal U}^\alpha  = U^\alpha - \sqrt{2}\, \bar\eta \, \Psi^\alpha + i \eta \bar \eta\, \dot{U}^\alpha \,
\end{equation}
where
\be\label{Z-and-U}
\begin{array}{rll}
&Z^a = z^a + \sqrt{2} \theta \phi^a - i \theta \bar \theta \dot{z}^a \qquad & (\bar D_\theta Z^a = 0 )\, ,  \\ [6pt]
&U^\alpha= u^\alpha + \sqrt{2}\, \theta\,\rho^\alpha -i\theta\bar\theta\,\dot u^\alpha \qquad &(\bar D_\theta U^\alpha = 0)
\end{array}
\ee
are the usual chiral  superfields and
\be
\label{Phi-Psi}
\begin{array}{rll}
&\Phi^a =  \varphi^a + \sqrt{2}\, \theta\,A^a -i\theta\bar\theta\,\dot \varphi^a\,,\qquad &\bar D_\theta\,{\Phi}^a=0\, \\ [6pt]
&\Psi^\alpha =  \varrho^\alpha + \sqrt{2}\, \theta\,B^\alpha -i\theta\bar\theta\,\dot \varrho^\alpha\,,\qquad &\bar D_\theta\,{\Psi}^\alpha=0\,.
\end{array}
\ee
are the chiral superfields of  type ({\bf 0}, {\bf 2}, {\bf 2}).
   In \eqref{Z-and-U}, \eqref{Phi-Psi},  the dynamical fields $z^a$, $u^\alpha$ and the complex auxiliary fields $A^a$, $B^\alpha$ are bosonic
whereas $\phi^a$, $\varphi^a$, $\rho^\alpha$, $\varrho^\alpha$ are fermionic.

The standard action of the twisted K\"ahler sigma model derived in \cite{GHR} is described by the action
$\sim {\int}dt\, d\theta d\bar\theta\, d\eta d\bar\eta\, {\cal K}({\cal Z},\bar{\cal Z},{\cal U},\bar{\cal U})$.
One can add to this expression  holomorphic $F$-terms  as in \p{act242-HKT}  and write
\be\label{act242}
\begin{array}{rcl}
S^{\rm bi-K} &=&  {\displaystyle \frac14}\displaystyle{\int}dt\, d\theta d\bar\theta\,
d\eta d\bar\eta\, {\cal K}({\cal Z},\bar{\cal Z},{\cal U},\bar{\cal U}) \\ [7pt]
&& -  {\displaystyle \frac12} \left[  \displaystyle{\int}dt\, d\theta \, d\eta \,
{\cal A}_{a}({\cal Z}) \dot{\cal Z}^{a}  + {\rm c.c.} \right]
  +{\displaystyle \frac12} \left[ \displaystyle{\int}dt\, d\theta \, d\bar\eta \,
{\cal B}_{\alpha}({\cal U})\,\dot{\cal U}^{\alpha} +  {\rm c.c.} \right]
\end{array}
\ee
Bearing in mind \p{calZa}, we can  integrate it over $d\eta$ and  $d\bar\eta$ and express the action in terms
of the ${\cal N} = 2$ superfields
({\bf 2}, {\bf 2}, {\bf 0}) and ({\bf 2}, {\bf 2}, {\bf 0}).  We obtain
\be
\lb{act422-2}
\begin{array}{rcl}
S^{\rm bi-K} &=&{\displaystyle \frac14}\displaystyle{\int}dt\, d\theta d\bar\theta\, \left[  \kappa_{a\bar b}
\left( D{Z}^{a}\bar D{\bar Z}^{\bar b}-2{\Phi}^{a}{\bar\Phi}^{\bar b} \right)
+  \kappa_{\alpha\bar \beta} \left( D{U}^{\alpha}\bar D{\bar U}^{\bar\beta}-2{\Psi}^{\alpha}{\bar\Psi}^{\bar \beta} \right) \right]\\  [11pt]
&&
-{\displaystyle \frac12}\displaystyle{\int}dt\, d\theta d\bar\theta\, \left[(\partial_{a}\partial_{\beta}{\cal K})\,{\Phi}^{a}{\Psi}^{\beta}-
(\bar\partial_{\bar a}\bar \partial_{\bar\beta}{\cal K})\,\bar{\Phi}^{\bar a}{\bar\Psi}^{\bar\beta} \right]\\  [11pt]
&&
+{\displaystyle \frac{1}{\sqrt{2}}}\displaystyle{\int}dt\, d\theta \left(  {\cal F}_{ab} (Z) \dot{Z}^{a}{\Phi}^{b}+
 {\cal G}_{\alpha\beta} (U) \dot{U}^{\alpha}{\Psi}^{\beta} \right)\\  [11pt]
&&
+{\displaystyle \frac{1}{\sqrt{2}}}\displaystyle{\int}dt\, d\bar\theta \left(  \bar{\cal F}_{\bar a\bar b} (\bar Z)
\dot{\bar Z}^{\bar a}\bar{\Phi}^{\bar b}+
 \bar{\cal G}_{\bar\alpha\bar\beta}(\bar U)
\dot{\bar U}^{\bar\alpha}\bar{\Psi}^{\bar\beta}\right)\, ,
\end{array}
\ee
where
\be
\lb{kap-FG}
\begin{array}{c}
 \kappa_{a\bar b} =  \partial_{a}\bar\partial_{\bar b}{\cal K}\,,\qquad
\kappa_{\alpha\bar \beta} =  -\partial_{\alpha}\bar\partial_{\bar \beta}{\cal K}  \\ [7pt]
 {\cal F}_{ab}  \ =\ \partial_a {\cal A}_b -  \partial_b {\cal A}_a, \qquad
{\cal G}_{\alpha\beta}  \ =\ \partial_\alpha {\cal B}_\beta -  \partial_\beta {\cal B}_\alpha \, .
\end{array}
  \ee

The bosonic part of the corresponding component Lagrangian  reads
\be
\lb{lagr242-b-Ka-2}
\begin{array}{rcl}
L^{\rm bi-K}_{bos} &=& \kappa_{a\bar b}\left(\dot z^{a} \dot {\bar z}^{\bar b} \ +\
A^{a} {\bar A}^{\bar b}\right)+ \kappa_{\alpha\bar \beta}\left(\dot u^{\alpha} \dot {\bar u}^{\bar\beta} \ +\
B^{\alpha} {\bar B}^{\bar\beta}\right)
\\ [11pt]
&& +\,{\cal F}_{ab}(z)\,\dot{z}^{a} A^{b}+
\bar{\cal F}_{\bar a\bar b}(\bar z)\,\dot{\bar z}^{\bar a}\bar A^{\bar b}
+ {\cal G}_{\alpha\beta}(u)\,\dot{u}^{\alpha} B^{\beta}+ \bar{\cal G}_{\bar\alpha\bar\beta}(\bar u)\,\dot{\bar u}^{\bar\alpha}\bar B^{\bar\beta}\, .
\end{array}
\ee
 The fermion terms are written in \p{lagr242-2f-add}, \p{lagr242-4f-add}.

Let us make now the following remark. Omitting the terms $\propto u, B$, we obtain the quasicomplex K\"ahler model with
the bosonic Lagrangian
  \be
L_{bos} = \kappa_{a\bar b}\left(\dot z^{a} \dot {\bar z}^{\bar b} \ +\
A^{a} {\bar A}^{\bar b}\right)
 +\,{\cal F}_{ab}(z)\,\dot{z}^{a} A^{b}+
\bar{\cal F}_{\bar a\bar b}(\bar z)\,\dot{\bar z}^{\bar a}\bar A^{\bar b}\,.
  \ee
If excluding the auxiliary fields $A^a, \bar A^{\bar b}$, we obtain the Lagrangian
 \be
\lb{bez-A}
L_{bos}  \to \dot{z}^{a} \dot {\bar z}^{\bar b} \left[ \kappa_{a\bar b} - F_{af} (\kappa^{-1})^{\bar c f} F_{\bar{c} \bar{b}} \right].
 \ee
We observe now that the actual metric in \p{bez-A} is not represented as $\partial_a \bar \partial_{\bar b} Q$ and is thus {\it not} K\"ahler (cf.  footnote 4).

The action \p{act242} is a restricted choice. In the bi-K\"ahler case, a more general form of the extra terms is possible. One can write, instead of $\int dt\, d\theta \, d\eta \,{\cal A}_{a}({\cal Z}) \dot{\cal Z}^{a}$, a linear combination of the expressions
$
\int dt\, d\theta  d\eta d\bar\eta \,
{\cal A}_{a}^{(1)}({\cal Z}, {\cal U}) D_\eta {\cal Z}^{a}
$
and
$
 \int dt\, d\theta  d\bar\theta d\eta \,
{\cal A}_{a}^{(2)}({\cal Z}, \bar{\cal U}) D_\theta {\cal Z}^{a}
$
with two different "vector potentials" depending also on the mirror superfields.
By the same token, the structure ${\cal B}_\alpha(U)$ is splitted in two:
${\cal B}^{(1)}_\alpha(U, {\cal Z})$ and ${\cal B}^{(2)}_\alpha(U, \bar {\cal Z})$.

To be more precise, the action
\be
\label{moreKah2}
\begin{array}{rcl}
{\tilde S}^{\rm bi-K} &=&  {\displaystyle \frac14}\displaystyle{\int}dt\, d\theta d\bar\theta\,
d\eta d\bar\eta\, {\cal K}({\cal Z},\bar{\cal Z},{\cal U},\bar{\cal U}) \\ [8pt]
&& -  {\displaystyle \frac{i}{4}} \left[  \displaystyle{\int}dt\, d\theta \, d\eta d\bar\eta \,
{\cal A}_{a}^{(1)}({\cal Z}, {\cal U})\, D_\eta {\cal Z}^{a}  - \displaystyle{\int}dt\, d\theta d\bar\theta \, d\eta \,
{\cal A}_{a}^{(2)}({\cal Z}, \bar{\cal U}) \,D_\theta {\cal Z}^{a}\, +\,  {\rm c.c.} \right]\\ [9pt]
&& +  {\displaystyle \frac{i}{4}} \left[  \displaystyle{\int}dt\, d\theta \, d\eta d\bar\eta \,
{\cal B}_{\alpha}^{(1)}({\cal U}, {\cal Z})\, \bar D_\eta {\cal U}^{\alpha}  - \displaystyle{\int}dt\, d\theta d\bar\theta \, d\bar\eta \,
{\cal B}_{\alpha}^{(2)}({\cal U}, \bar{\cal Z}) \,D_\theta {\cal U}^{\alpha}\, +\,  {\rm c.c.} \right]
\end{array}
\ee
produces the bosonic  Lagrangian \p{lagr242-b-Ka-2} where the field strengths ${\cal F}_{ab}$ and ${\cal G}_{\alpha\beta}$
are defined by the same expressions \p{kap-FG} with
\be
\lb{gen-AB}
{\cal A}_a ={\cal A}_{a}^{(1)}(z, u)+{\cal A}_{a}^{(2)}(z, \bar u), \qquad
{\cal B}_\alpha = {\cal B}_{\alpha}^{(1)}(u, z)+{\cal B}_{\alpha}^{(2)}(u, \bar z)\, .
\ee

Eq.\p{moreKah2} defines a generic quasicomplex bi-K\"ahler model.

\vspace{2mm}

\centerline{\it Hamiltonian reduction.}

\vspace{2mm}

We will now show  how the bi-K\"ahler models \p{act242}, \p{moreKah2}
are obtained from the bi-HKT model \p{act4} by Hamiltonian reduction.

Consider first the pure HKT model. The bosonic part of its Lagrangian is given by
\be
\lb{lagr4-b-hkt}
L^{\rm HKT}_{bos} =
h^{ab}_{m\bar n}\, \dot v_a^m \dot {\bar v}_b^{\bar n} \, ,\qquad
h^{ab}_{m\bar n}=\left(  \frac{\partial^2 }{\partial v^m_a \partial \bar v^{\bar n}_b }   +
\epsilon_{m k} \epsilon_{\bar n \bar l} \frac{\partial^2 }{\partial \bar v^{\bar k}_a \partial v^l_b}\right){\cal L}(v,\bar v) \, .
\ee
It is convenient to introduce the notations $z^a = v^1_a, \ \xi^a = \bar v^2_a$.
Then
 \be
\lb{Lb-HKT-xi}
L^{\rm HKT}_{bos} &=& \kappa_{a\bar b}\left(\dot z^{a} \dot {\bar z}^{\bar b} +
\dot \xi^{a} \dot {\bar \xi}^{\bar b}\right)\ +\
 \,{\cal F}_{ab}\,\dot{z}^{a} \dot \xi^{b}+
\bar{\cal F}_{\bar a\bar b}\,\dot{\bar z}^{\bar a}\dot {\bar\xi}^{\bar b} \, ,
\ee
where
\be
\lb{h-red-hkt-kap}
\kappa_{a\bar b}=h^{a  b}_{1 \bar 1}=h^{b  a}_{2 \bar 2}=\left(  \frac{\partial^2 }{\partial z^a \partial \bar z^{\bar b} }   +
\frac{\partial^2 }{\partial \xi^a \partial \bar \xi^{\bar b} }\right){\cal L}\,,
\ee

\be
\lb{h-red-hkt-F}
{\cal F}_{ab} =  h^{ab}_{1 \bar 2}=\left( \frac{\partial^2 }{\partial z^{a} \partial \xi^{b } } - \frac{\partial^2 }{\partial z^{b} \partial \xi^{a } } \right) {\cal L}\,,\quad
\bar{\cal F}_{\bar a\bar b}=  h^{ab}_{2 \bar 1}=\left( \frac{\partial^2 }{\partial \bar z^{\bar a} \partial \bar \xi^{\bar b} } -
\frac{\partial^2 }{\partial \bar z^{\bar b} \partial \bar \xi^{\bar a} } \right) {\cal L}\,,
\ee

For the model to be reducible, $\kappa_{a \bar b}$ and $F_{ab}$ should not depend on a half of coordinates, which we choose here to be $\xi^a$. This imposes restrictions for the prepotential.
A generic {\it relevant}
expression for the prepotential
satisfying this  condition  reads
    \be
\lb{h-red-hkt2}
{\cal L}={\cal K}({\cal Z},\bar {\cal Z})+ {\cal A}_a({\cal Z})\,\Xi^a+ \bar{\cal A}_{\bar a}(\bar {\cal Z})\,\bar\Xi^{\bar a}\, ,
   \ee
 with real ${\cal K}({\cal Z}, \bar {\cal Z})$ ( ${\cal Z}^a ={\cal V}^1_a,  \, \Xi^a = {\cal V}^2_a$).
  Then
$\kappa_{a\bar b} = \partial_a \bar \partial_{\bar b} {\cal K}$ and
     ${\cal F}_{ab} = \partial_a {\cal A}_b - \partial_b {\cal A}_a $.

Saying ``relevant'', we meant the following. The problem has a gauge freedom. One can generalize \p{h-red-hkt2} by
 adding in the prepotential the terms like
$$
\Delta {\cal L} \ =\ {\cal B}_a({\cal Z}) \bar \Xi^a + {\rm c.c.} \ \ \ \ \ {\rm or} \ \ \ \ \
\Delta {\cal L} \ =\ Q({\cal Z}, \bar {\cal Z}) - (\partial_a \partial_{\bar b} Q) \Xi^a \bar\Xi^{\bar b} \, ,
  $$
which do not contribute, however, in the component Lagrangian. Also one can add to
${\cal A}_a(Z)$ the gradient of an arbitrary holomorphic function with the same effect (or rather its absence).
Such
 generalizations are thus {\it irrelevant}.

Note that the bosonic kinetic part \p{Lb-HKT-xi} of the Lagrangian involves, besides
the term $\kappa_{a\bar b} \dot z^a \dot {\bar z}^{\bar b}$, also the
term $\kappa_{a\bar b}  \dot \xi^a \dot {\bar \xi}^{\bar b}$ with the time derivatives of
$\xi^a$, even in the case when the prepotential
in \p{act1} depends only on ${\cal V}^1_a$. This is due to the fact that the component expansion of
${\cal V}^1_a$ involves also $\bar V^2_a$, as is clearly seen from
\p{V-4-2exp}.

 As was noticed in \cite{BerPash,GatesRana,BKMO} and  discussed in detail in \cite{commentHKT}, under reduction,
the time derivatives of the reduced coordinates go over into the auxiliary fields of the reduced model.
In our case, $\dot{\xi}^a  \to A^a$.
We thus reproduce the expression \p{bez-A} containing a half of the terms in \p{lagr242-b-Ka-2}
involving the fields from the ordinary multiplets.
The fermion terms are restored in both original and reduced model by ${\cal N} =4$ supersymmetry and also go one into another.

This proves {\bf Proposition 5}.

\vspace{.1cm}

We are ready now to prove the theorem:

\begin{thm}
 Consider  a bi-HKT model \p{act4} with the prepotential
\be
\lb{L-prep2}
\begin{array}{lcl}
{\cal L} &=& {\cal K}({\cal Z},\bar {\cal Z}; {\cal U}, \bar {\cal U}) + \Big[{\cal A}^{(1)}_a({\cal Z}, {\cal U}) +  {\cal A}^{(2)}_a({\cal Z}, \bar {\cal U}) \Big]  \Xi^a+
\Big[\bar {\cal A}^{(1)}_a(\bar {\cal Z}, \bar {\cal U}) +  \bar {\cal A}^{(2)}_a( \bar {\cal Z},  {\cal U}) \Big]  \bar\Xi^{\bar a} \\ [9pt]
&&\qquad\qquad\qquad\, - \Big[ {\cal B}^{1}_\alpha({\cal U}, {\cal Z})    + {\cal B}^{2}_\alpha({\cal U}, \bar {\cal Z})  \Big]
\Sigma^\alpha  -      \Big[ \bar {\cal B}^{1}_\alpha(\bar {\cal U}, \bar {\cal Z})    + \bar {\cal B}^{2}_\alpha( \bar {\cal U},  {\cal Z})  \Big]   \bar\Sigma^{\bar \alpha}
\end{array}
\ee
(${\cal Z}^a = {\cal V}^1_a, \Xi^a = {\cal V}^2_a, {\cal U}^\alpha = {\cal W}^1_\alpha, \Sigma^\alpha = {\cal W}^2_\alpha$).
It is reducible with respect to the coordinates $\xi^a, \sigma^\alpha$ and gives after reduction the
quasicomplex bi-K\"ahler model \p{moreKah2}.
\end{thm}

\begin{proof}
It follows closely the proof of {\bf Proposition 5}, one should only take into consideration the mirror sector. The bosonic Lagrangian  \p{lagr4-b-app} now reads
\be
\lb{lagr4-b-HKT2}
\begin{array}{rcl}
L^{bi-HKT}_{b} &=& \kappa_{a\bar b}\left(\dot z^{a} \dot {\bar z}^{\bar b} \ +\
\dot \xi^{a} \dot {\bar \xi}^{\bar b}\right)+ \kappa_{\alpha\bar \beta}\left(\dot u^{\alpha} \dot {\bar u}^{\bar\beta} \ +\
\dot \sigma^{\alpha} \dot {\bar \sigma}^{\bar\beta}\right)
\\ [11pt]
&& +\,{\cal F}_{ab}\,\dot{z}^{a} \dot \xi^{b}+
\bar{\cal F}_{\bar a\bar b}\,\dot{\bar z}^{\bar a}\dot {\bar\xi}^{\bar b}
+ {\cal G}_{\alpha\beta}\,\dot{u}^{\alpha} \dot \sigma^{\beta}+ \bar{\cal G}_{\bar\alpha\bar\beta}\,\dot{\bar u}^{\bar\alpha}\dot {\bar \sigma}^{\bar\beta}\, ,
\end{array}
\ee
with
  \be
\lb{kap-FG-1}
\begin{array}{rcl}
\kappa_{a\bar b}  \ = \
{\displaystyle\frac {\partial^2}{\partial z^a \partial \bar z^{\bar b}}\,  {\cal K}}\, , &\quad &
{\displaystyle\kappa_{\alpha\bar\beta}  \ = \
-\frac {\partial^2}{\partial u^\alpha \partial \bar u^{\bar\beta} }\, {\cal K}\,,} \\ [10pt]
{\displaystyle  {\cal F}_{ab} \ =\ \frac {\partial}{\partial z^a} \left[ {\cal A}^{(1)}_b + {\cal A}^{(2)}_b\right]
- (a \leftrightarrow b)\,,} &\quad  &
{\displaystyle{\cal G}_{\alpha\beta} \ =\ \frac {\partial}{\partial u^\alpha} \left[ {\cal B}^{(1)}_\beta + {\cal B}^{(2)}_\beta\right]
- (\alpha \leftrightarrow \beta)\,.}
\end{array}
\ee

The metric \p{kap-FG-1} does not depend on $\xi^a, \sigma^\alpha$, and one can perform
the reduction with respect to these variables.
After that, the derivatives $\dot{\xi}^a, \dot{\sigma}^\alpha$ go over into auxiliary fields $A^a, B^\alpha$ and we reproduce the reduced bosonic Lagrangian \p{lagr242-b-Ka-2} with
${\cal F} \to  {\cal F}^{(1)} +  {\cal F}^{(2)} $ , ${\cal G} \to  {\cal G}^{(1)} +  {\cal G}^{(2)} $
of the quasicomplex bi-K\"ahler model \p{moreKah2}.
The fermion terms are restored by supersymmetry.
We amused ourselves to check it explicitly for the restricted Ansatz,
 \be
\lb{L-prep2-Ans}
{\cal L} = {\cal K}({\cal Z},\bar {\cal Z}; {\cal U}, \bar {\cal U}) + {\cal A}_a({\cal Z})\,  \Xi^a +
\bar {\cal A}_a(\bar {\cal Z})\, \bar\Xi^{\bar a}
-  {\cal B}_\alpha({\cal U})\,
\Sigma^\alpha  -       \bar {\cal B}_\alpha(\bar {\cal U})  \,  \bar\Sigma^{\bar \alpha} \, .
\ee
 In this
case,  the reduction of \p{lagr4-2f-app} and \p{lagr4-4f-app} gives \p{lagr242-2f-add} and \p{lagr242-4f-add}  under the identifications
 \be
\begin{array}{c}
\psi_a^1=\phi^a\,,\quad \bar\psi_a^1=\bar\phi^{\bar a}\,,\qquad
\psi_a^2= -i\bar\varphi^{\bar a}\,,\quad \bar\psi_a^2= i\varphi^a\,, \\[7pt]
\chi_\alpha^1=\rho^\alpha\,,\quad \bar\chi_\alpha^1=\bar\rho^{\bar \alpha}\,,\qquad
\chi_\alpha^2= -i\bar\varrho^{\bar \alpha}\,,\quad \bar\chi_\alpha^2= i\varrho^\alpha
\,.
\end{array}
\ee
\end{proof}

\subsection{({\bf 2}, {\bf 4}, {\bf 2}) = ({\bf 1}, {\bf 2}, {\bf 1})$\oplus$({\bf 1}, {\bf 2}, {\bf 1})}

Alternatively, one can represent a superfield \p{242-sf} as a couple of real  ({\bf 1}, {\bf 2}, {\bf 1}) ${\cal N} = 2$ superfields $X^m_a$ and a superfield \p{242-sf-tw} as a couple of real ({\bf 1}, {\bf 2}, {\bf 1}) superfields
$Y^\mu_\alpha$ and express the action in terms of  $X^m_a$ and $Y^\mu_\alpha$.
To find this expression, it is convenient to perform the Hamiltonian reduction of the original bi-HKT action in a different way, getting rid not of the complex variables $v^2_a$ and $w^2_\alpha$, but of the imaginary parts, Im$(v^m_a)$ and Im$(w^\mu_\alpha)$.

The result of such a reduction is, of course, the same as for the reduction considered above. To see that explicitly, one can write
 \be
\label{v-xlam}
v^m_a = x^m_a + i\lambda^m_a\,, \qquad w^\mu_\alpha = y^\mu_\alpha + i t^\mu_\alpha
 \ee
and introduce new complex variables,
  \be
\label{z-x12}
z^a = x^1_a + ix^2_a\,, \quad \xi^a = \lambda^1_a + i \lambda^2_a\,, \qquad u^\alpha = y^1_\alpha + iy^2_\alpha\,,
\quad \sigma^\alpha = t^1_\alpha + it^2_\alpha \, .
 \ee
We choose, again, the prepotential in the form \p{L-prep2}, perform the reduction with respect to $\xi^a, \sigma^\alpha$ and
obtain \p{moreKah2}.

Consider now  the equivalent form \p{act4-2}, \p{lagr4-2} of the bi-HKT action when the latter is expressed via ${\cal N} = 2$
chiral superfields. As was mentioned, this action belongs to the class of Dolbeault SQM models modified by the inclusion of
holomorphic torsions [the second line in \p{lagr4-2}]. The Hamiltonian reduction of the first line was studied in
\cite{quasi}. It gives a quasicomplex de Rham model with the Lagrangian
  \be
\label{1-line}
{\cal L}^{\rm first\ line} \ =\  \frac 12 \left[g^{ab}_{mn} + i b^{ab}_{mn} \right] D X^m_a \bar D X^n_b +
\frac 12 \left[g^{\alpha \beta}_{\mu\nu}  + i b^{\alpha\beta}_{\mu\nu} \right]
D Y^\mu_\alpha \bar D Y^\nu_\beta \, ,
 \ee
where $X^m_a$ and $Y^m_\alpha$ are real ${\cal N} = 2$ superfields, and $g^{ab}_{mn}, g^{\alpha \beta}_{\mu\nu}$ and
$b^{ab}_{mn}, b^{\alpha\beta}_{\mu\nu} $ are, correspondingly, the
real symmetric and imaginary antisymmetric parts of the Hermitian metrics $h^{ab}_{m\bar n}, h^{\alpha  \beta}_{\mu \bar\nu}$ entering  \p{lagr4-2}, multiplied by $1/2$, as dictated by \p{metr-red} . 
If choosing the restricted ansatz \p{L-prep2-Ans} and expressing the complex superfields via the real ones,
we obtain the following explicit expressions for the metric:
  \be
\lb{g-sym}
g^{ab}_{mn}  \ =\  \frac 12 \Big(\partial^a_m  \partial_{n}^{b}    +
\epsilon_{m k} \epsilon_{n l}  \partial^a_{k}  \partial^b_l \Big)  \, {\cal K}(Z, \bar Z, U, \bar U)
\ee
and
\be
\lb{b-antisym}
b^{ab}_{mn} = - { {\cal F}}^{ab}_{mn} \ =\  \partial^b_n   { {\cal A}}{}^{a}_m     -
\partial^a_m   { {\cal A}}{}^{b}_n \, ,
\ee
with
\be
\lb{A-A}
{ {\cal A}}{}^{a}_1(X) = {\cal A}_{a}(Z)+ \bar {\cal A}_{a}(\bar Z)  \,,\qquad
{{\cal A}}{}^{a}_2(X)  =  i\left[{\cal A}_{a}(Z)-\bar{\cal A}_{a}(\bar Z)\right]
\ee
being the real 2-dimensional "vector-potentials".
The derivatives $\partial^a_m$, $\partial^\alpha_\mu$ stand now for 
$\partial/\partial X_a^m$, $\partial/\partial Y_\alpha^\mu$.

Similarly,
\be
\lb{g-sym-mir}
g^{\alpha \beta}_{\mu\nu}  \ =\  -\frac 12 \Big(\partial^\alpha_\mu  \partial_\nu^\beta    +
\epsilon_{\mu \kappa} \epsilon_{\nu \lambda}  \partial^\alpha_\kappa  \partial^\beta_\lambda \Big)  \, {\cal K} (Z, \bar Z, U, \bar U)
\ee
and
\be
\lb{b-antisym-mir}
b^{\alpha\beta}_{\mu\nu}  = - { {\cal G}}^{\alpha\beta}_{\mu\nu}  \ =\  \partial^\beta_\nu   { {\cal B}}{}^{\alpha}_\mu     - \partial^\alpha_\mu   { {\cal B}}{}^{\beta}_\nu
\ee
with
\be
\lb{B-B}
{ {\cal B}}{}^{\alpha}_1(Y) = {\cal B}_{\alpha}(U)+ \bar {\cal B}_{\alpha}(\bar U) \,, \qquad
{ {\cal B}}{}^{\alpha}_2(Y) = i \left[ {\cal B}_{\alpha}(U) -  \bar {\cal B}_{\alpha}(\bar U) \right] \,.
\ee

The reduction of the second line in \p{lagr4-2} is also easily performed, using the results of  Ref.\cite{FIS1}
and bearing in mind that the terms linear in $\Xi^a$ and
$\Sigma^\alpha$ in the Ansatz \p{L-prep2-Ans} do not contribute there.

The explicit component expression for the 2-fermion term $\sim \dot{z} \psi \psi$ in the Lagrangian
derived from the generic
holomorphic contribution to the Dolbeault action is
\be
 S^{\rm hol}_{\rm Dolb}  = \frac 14 \int d^2\theta \, {\cal B}_{MN} (Z, \bar Z) \, D Z^M D Z^N + {\rm c.c.}  \nn
\Downarrow \qquad\qquad\qquad\qquad\qquad \nn
\left[-3i \partial_{[M}  {\cal B}_{NP]} \dot{z}^M \psi^N \bar \psi^{\bar P}  + {\rm c.c.} \right]
 +  {\rm other\ terms}. \lb{hol-Dol}
 \ee
On the other hand,    the explicit component expression for the 2-fermion term $\sim \dot{x} \psi \psi$ in the Lagrangian
derived from the generic holomorphic contribution to the de Rham action is
\be
 S^{\rm hol}_{\rm Rham}  =  \int d^2\theta \, {\cal B}_{MN} (X) \, D X^M D X^N + {\rm c.c.}
\nn
\Downarrow \qquad\qquad\qquad\qquad\qquad \nn
\left[-3i \partial_{[M}  {\cal B}_{NP]} \dot{x}^M \psi^N \bar \psi^{ P}  + {\rm c.c.} \right]
 + {\rm other\ terms}. \lb{hol-Rham}
\ee
Under reduction, \p{hol-Dol} goes over to \p{hol-Rham}.
The factor $  {\frac 14}$ is compensated
roughly by the same mechanism as for the kinetic term:
$  {\frac 14} \int d^2\theta DZ \bar D \bar Z = \dot{z} \dot{\bar z} + \cdots$ gives
 $ \int d^2\theta DX \bar D X = \dot{x} \dot{x} + \cdots$ after reduction.

  We finally obtain the action
  \be
\lb{N=2_XY}
S = \frac 12  \int dt\,  d\theta d\bar\theta \,\Big[ \left(g^{ab}_{mn} + i b^{ab}_{mn} \right) D X^m_a \bar D X^n_b +
\left(g^{\alpha \beta}_{\mu\nu} + i b^{\alpha\beta}_{\mu\nu} \right)
D Y^\mu_\alpha \bar D Y^\nu_\beta   \nn
+ \,\frac 12\, \epsilon_{mn} \epsilon_{\mu\nu} (\partial^a_n \partial^\alpha_\nu {\cal K} )
\left( \bar D X^m_a \bar D Y^\mu_\alpha - D X^m_a  D Y^\mu_\alpha \right) \Big] \,.
 \ee
The extra additional ${\cal N} =2$ supersymmetry present in the action \p{N=2_XY} is realized as
\be
\lb{extra-XY}
\begin{array}{ll}
&\delta X_a^m = \ -\varepsilon_\eta \epsilon^{mn} \bar D X_a^n + \bar \varepsilon_\eta \epsilon^{mn} DX_a^n \, , \\ [8pt]
&\delta Y_\alpha^\mu = \ - \bar \varepsilon_\eta \epsilon^{\mu\nu} \bar D Y_\alpha^\nu + \varepsilon_\eta \epsilon^{\mu\nu} DY_\alpha^\nu \, .
\end{array}
\ee

When checking explicitly the invariance of the action, it is convenient to transform the  expressions  \p{g-sym}, \p{g-sym-mir} to the form
\begin{equation}
g^{ab}_{mn}   =  \frac 12 \Big(\partial^a_k  \partial_{k}^{b}{\cal K} \Big)\, \delta_{mn}
+\frac 12 \Big(\epsilon_{kl}  \partial^a_{k}  \partial^b_l {\cal K}\Big)\,  \epsilon_{mn}\,,
\qquad
g^{\alpha \beta}_{\mu\nu}   =  -\frac 12 \Big(\partial^\alpha_\lambda  \partial_\lambda^\beta{\cal K}\Big)\,\delta_{\mu\nu}
- \frac 12 \Big(\epsilon_{\lambda\rho}  \partial^\alpha_\lambda  \partial^\beta_\rho  {\cal K}\Big)\,   \epsilon_{\mu\nu}
\end{equation}
and use the d'Alembert-Euler analyticity conditions
\begin{equation}
\partial^{a}_m{\cal A}{}^{b}_m = 0  \,,\qquad
\epsilon_{mn}\partial^{a}_m{\cal A}{}^{b}_n = 0\,,
\qquad\qquad
\partial^{\alpha}_\mu {\cal B}{}^{\beta}_\mu = 0 \,, \qquad
\epsilon_{\mu\nu}\partial^{\alpha}_\mu {\cal B}{}^{\beta}_\nu = 0
\end{equation}
for the real potentials \p{A-A}, \p{B-B}.

Transformations \p{extra-XY} have the same form as in \p{SUSYtran-N=2hid-VW} by  replacing $V^m$, $\bar V^m \to X^m$ and
$W^\mu$, $\bar W^\mu \to Y^\mu$.

Introducing now $$\varepsilon = \varepsilon_+ + i \varepsilon_-\, , \ \ \ \ \ \ \ \ \ \ \
D = D_+ + i D_-$$ with real $\varepsilon_\pm, D_\pm$, one can rewrite \p{extra-XY} in the form
  \be
\delta \left( \begin{array}{c} X \\ Y \end{array} \right)^M \ =\ 2i\varepsilon_+ \, {J}^M{}_N   D_-    \left( \begin{array}{c} X \\ Y \end{array} \right)^N - 2i\varepsilon_-  \, {I}^M{}_N D_+    \left( \begin{array}{c} X \\ Y \end{array} \right)^N \, ,
  \ee
where the matrices $I, J$ coincide with the matrices ${\cal I}, {\cal J}$  written in \p{comp-str-B-c}. Only now they have the meaning of two different {\it complex} structures rather than the hypercomplex structures (as was the case for the bi-HKT manifolds).

{\bf Remark 3}. As we have seen, the  quasicomplex structures $\propto b_{mn}^{ab}, \ b_{\mu\nu}^{\alpha\beta}$ in the ${\cal N} = 2$ action \p{N=2_XY} appear only in the models involving extra holomorphic terms $\propto {\cal A}, {\cal B}$ in \p{moreKah2}. These terms are specific for SQM, they cannot be obtained from a Lorentz-invariant
2-dimensional field theory by dimensional reduction. However, they can be derived starting from certain Lorentz-noninvariant $2d$ sigma models \cite{CKT,quasi}.

\subsection{Supercharges}

The simplest way to derive the supercharges in the bi-K\"ahler model is to use the  expressions \p{S-CKT}, \p{R-CKT}, \p{S-B-c} for the  ({\bf 4}, {\bf 4}, {\bf 0})  supercharges, derived above, and perform the Hamiltonian reduction.
The expressions thus obtained have exactly the same form as \p{S-CKT}, \p{R-CKT}, \p{S-B-c} simplified by the fact that
the derivatives $\partial_{\cal M}, \bar \partial_{\bar {\cal M}}$ and the momenta $\Pi_{\cal M}, \bar \Pi_{\bar {\cal M}}$ are not distinguished anymore and go to the real momenta and derivatives. The substitution rules are the following,
   \be
\partial_{\cal M}, \bar \partial_{\bar {\cal M}} \to \ \frac 12\, \partial_M\,,  \qquad \Pi_{\cal M}, \bar \Pi_{\bar {\cal M}} \to \ \frac 12\, \Pi_M\,,
\qquad \psi^{\cal M} \to \sqrt{2}\, \psi^M\,.
 \ee
 (The last rule follows from the presence of the factor $\sqrt{2}$ in the component expansion \p{chir-comp} of the chiral
multiplet and its absence in the conventionally defined component expansion for the ({\bf 1}, {\bf 2}, {\bf 1}) multiplet,
$X^M = x^M + \theta \psi^M + \bar \psi^M \bar \theta + F^M \theta \bar \theta$.)

We derive,
\begin{equation}\label{SR-HKT-red}
\begin{array}{rcl}
S&=&\psi_a^m\left(\Pi^a_m- \frac i2 \,\partial^a_{m}h^{bc}_{k\bar l}\, \psi_b^k\bar\psi_{c}^{l}
- \frac i2 \,\partial^a_{m}h^{\beta\gamma}_{\nu\bar \lambda}\,
\chi_\beta^\nu\bar\chi_{\gamma}^{\lambda} +  \frac i2 \,\epsilon_{mn}\epsilon_{\nu\lambda} \, \partial^a_{n}h^{\alpha\beta}_{\mu\bar \lambda}\,
\chi_\alpha^\mu\chi_{\beta}^{\nu}\right)
\\[9pt]
&&+\chi_\alpha^\mu\left(\Pi^\alpha_\mu- \frac i2 \,\partial^\alpha_{\mu}h^{\beta\gamma}_{\nu\bar \lambda}\,
\chi_\beta^\nu\bar\chi_{\gamma}^{\lambda}- \frac i2 \,\partial^\alpha_{\mu}h^{bc}_{k\bar l}\, \psi_b^k\bar\psi_{c}^{l}
 + \frac i2 \,\epsilon_{\mu\nu}
\epsilon_{nk}  \, \partial^\alpha_{\nu} h^{ab}_{m\bar k}\, \psi_a^m\psi_{b}^{n}\right),\\ [12pt]
\bar S&=&\bar\psi_a^{m}\left(\Pi^a_{m}+  \frac i2 \,\partial^a_{m}h^{bc}_{k\bar l}\, \psi_b^k\bar\psi_{c}^{l}
+ \frac i2 \,\partial^a_{m}h^{\beta\gamma}_{\nu\bar \lambda}\,
\chi_\beta^\nu\bar\chi_{\gamma}^{\lambda}
- \frac i2 \,\epsilon_{mn}\epsilon_{\mu\lambda} \,\partial^a_{n} h^{\alpha\beta}_{\lambda\bar \nu}\,
\bar\chi_\alpha^{\mu}\bar\chi_{\beta}^{\nu}\right)\\[9pt]
&&+\bar\chi_\alpha^{\mu}\left(\Pi^\alpha_{\mu}+  \frac i2 \,\partial^\alpha_{\mu}h^{\beta\gamma}_{\nu\bar \lambda}\,
\chi_\beta^\nu\bar\chi_{\gamma}^{\lambda}+ \frac i2 \,\partial^\alpha_{\mu}h^{bc}_{k\bar l}\, \psi_b^k\bar\psi_{c}^{l}
- \frac i2 \,\epsilon_{\mu\nu}
\epsilon_{mk}  \, \partial^\alpha_{\nu} h^{ab}_{k\bar n}\, \bar\psi_a^{m}\bar\psi_{b}^{n}\right) \,,
\end{array}
\end{equation}

\begin{equation}\label{SR-HKT-red-1}
\begin{array}{rcl}
R&=&
\psi_a^{n}\,\epsilon_{nm}\left(\Pi^a_{m}- \frac i2 \,\partial^a_{m}h^{b c}_{k\bar l}\, \psi_b^k\bar\psi_{ c}^{l}
+ \frac i2 \,\partial^a_{m}h^{\beta\gamma}_{\nu\bar \lambda}\,
\chi_\beta^\nu\bar\chi_{\gamma}^{\lambda}
- \frac i2 \,\epsilon_{mk}\epsilon_{\mu\lambda} \, \partial^a_{k} h^{\alpha\beta}_{\lambda\bar \nu}\,
\bar\chi_\alpha^{\mu}\bar\chi_{\beta}^{\nu}\right)\\[9pt]
&&-
\bar\chi_\alpha^{\nu}\,\epsilon_{\nu\mu}\left(\Pi^\alpha_\mu+  \frac i2\,\partial^\alpha_{\mu}h^{\beta\gamma}_{\nu\bar \lambda}\,
\chi_\beta^\nu\bar\chi_{\gamma}^{\lambda} -  \frac i2 \,\partial^\alpha_{\mu}h^{bc}_{k\bar l}\, \psi_b^k\bar\psi_{c}^{l}
 + \frac i2 \,\epsilon_{\mu\lambda}
\epsilon_{nk}  \, \partial^\alpha_{\lambda} h^{ab}_{m\bar k}\, \psi_a^m\psi_{b}^{n}\right)
\,,\\ [12pt]
\bar R&=&
\bar\psi_a^{n}\,\epsilon_{nm}\left(\Pi^a_m + \frac i2 \,\partial^a_{m}h^{bc}_{k\bar l}\, \psi_b^k\bar\psi_{c}^{l}
- \frac i2 \,\partial^a_{m}h^{\beta\gamma}_{\nu\bar \lambda}\,
\chi_\beta^\nu\bar\chi_{\gamma}^{\lambda}
 + \frac i2 \,\epsilon_{mk} \epsilon_{\nu\lambda}  \, \partial^a_{k} h^{\alpha\beta}_{\mu\bar \lambda}\,
\chi_\alpha^\mu\chi_{\beta}^{\nu}\right)
\\[9pt]
&& -
\chi_\alpha^\nu\,\epsilon_{\nu\mu}\left(\Pi^\alpha_{\mu} - \frac i2\,\partial^\alpha_{\mu}h^{\beta\gamma}_{\nu\bar \lambda}\,
\chi_\beta^\nu\bar\chi_{\gamma}^{\lambda} + \frac i2 \,\partial^\alpha_{\mu}h^{bc}_{k\bar l}\, \psi_b^k\bar\psi_{c}^{l}
- \frac i2 \,\epsilon_{\mu\lambda}
\epsilon_{mk}   \, \partial^\alpha_{\lambda} h^{ab}_{k\bar n}\, \bar\psi_a^{m}\bar\psi_{b}^{n}\right)
\,.
\end{array}
\end{equation}

More compact expressions can be obtained by the reduction of \p{S-B-c} and the similar expressions for $R, \bar R$.

\subsection{Geometry}

 Consider first an {\it ordinary} bi-K\"ahler manifold without extra gauge potentials in Eqs. \p{moreKah2}.
The complex structures $I,J$, given by \p{comp-str-B-c},  trivially satisfy  the condition \p{Nijen} and are integrable.
They are not covariantly constant with the ordinary Levi-Civita connection \p{Lev-Civ}, but both $I$ and $J$ are covariantly constant with torsionful connections. It was noticed \cite{GHR} that the following property holds,

\begin{prop} The Bismut torsions for the complex structures $I,J$ coincide modulo  sign,
    \be
\lb{C=-C}
 C^{(B)}_{MNK} (I) \ =\ -C^{(B)}_{MNK} (J) \, .
  \ee
\end{prop}

\begin{proof}
As the complex structures are integrable, the holomorphic coordinates can be chosen. Let us do it for the complex
structure $J$. The Bismut totally antisymmetric torsions \p{CBismut} can be expressed in these terms as \cite{IS}
\be
\lb{C-Bis-compl}
\begin{array}{rcl}
C_{{\cal M} {\cal N} \bar {\cal K}}(J) &=& \partial_{\cal N} {\cal H}_{{\cal M} \bar {\cal K}} -   \partial_{\cal M} {\cal H}_{{\cal N} \bar {\cal K}},
\\[9pt]
C_{\bar {\cal M} \bar {\cal N}  {\cal K}}(J) &=& \left[C_{{\cal M} {\cal N} \bar {\cal K}} (J) \right]^* \ =\ \partial_{\bar {\cal N}} {\cal H}_{{\cal K} \bar  {\cal M}} -
\partial_{\bar {\cal M}} {\cal H}_{{\cal K} \bar {\cal N} },
\end{array}
\ee
and those obtained from them by the permutation of the indices. The other components of $C_{MNK}$ vanish.

${\cal H}_{{\cal M} \bar {\cal N}}$ is the  Hermitian metric  tensor. For a bi-K\"ahler manifold, it has a block
diagonal form,
  \be
\lb{block}
  {\cal H}_{{\cal M} \bar {\cal N}} \ =\ {\rm diag}  \left( {\kappa }_{a\bar b}, {\kappa}_{\alpha\bar\beta} \right) \, .
   \ee
with
$\kappa_{a \bar b} = \partial_a \partial_{\bar b} {\cal K}$ and
$\kappa_{\alpha \bar \beta}  = -\partial_\alpha \partial_{\bar\beta} {\cal K}$.

 Then the torsion components $C_{ab\bar c}, C_{\bar a \bar b c}, C_{\alpha \beta \bar\gamma},
C_{\bar\alpha \bar\beta \gamma}$, associated with a single sector, vanish. There are nonzero  mixed components,
\be
\lb{C-mixed}
C_{a\alpha \bar b }(J) = \partial_\alpha \kappa_{a \bar b}, \ \ \ \ C_{\bar a \bar \alpha b}(J) = \partial_{\bar \alpha} \kappa_{b \bar a},
\ \ \ \ C_{ a  \alpha \bar \beta}(J) = -\partial_a \kappa_{\alpha  \bar \beta}, \ \ \ \ \
C_{ \bar a  \bar \alpha  \beta}(J) = -\partial_{\bar a} \kappa_{\beta  \bar \alpha},
  \ee
and those obtained by permutation.

Now, the complex structure $I$ differs from $J$ by a sign in the mirror sector. This simply means that the definitions of holomorphic and antiholomorphic coordinates are now interchanged. In other words, the Bismut torsions for the complex structure $I$ are obtained from \p{C-mixed} by interchanging $\alpha \leftrightarrow \bar\alpha,  \beta \leftrightarrow \bar\beta$. Using the antisymmetry of $C$, it is not difficult to see that this amounts to changing the sign, and the relation
\p{C=-C} holds.

\end{proof}

It does not quite work in the opposite direction. What one can prove is the following,

\begin{prop}
Consider a manifold having two commuting complex structures $I,J$  with opposite Bismut torsions, $C(I)= -C(J)$.
 Then the metric can be brought to the block diagonal form,
 \be
\lb{K-12}
 {\cal H}_{{\cal M} \bar {\cal N}} \ =\ {\rm diag}  \left( \partial_a \bar \partial_{\bar b} {\cal K}_1,
\partial_\alpha \bar \partial_{\bar\beta} {\cal K}_2 \right) \, .
\ee
\end{prop}

\begin{proof}
It is not difficult to show that a couple of  commuting complex structures can be brought to the canonical form
\p{comp-str-B-c}.
The holomorphic coordinates associated with the complex structure $I$ are obtained from those associated with the complex structure $J$ by  conjugating the coordinates in the mirror sector,
$u_J^\alpha \to \bar u_I^{\bar \alpha}$. The metric should be Hermitian both in terms of
$(z^a,  u_J^\alpha)$ and in terms of
$(z^a,  u_I^\alpha)$. That means that it is bound to have a block diagonal form,
      \be
\lb{metr-J}
 ds^2 \ =\ 2\, {\mathrm h}_{a \bar b} \,dz^a d\bar z^{\bar b} + 2 \, {\mathrm h}_{\alpha \bar \beta}\,
d u_J^\alpha d \bar u_J^{\bar\beta} \, .
 \ee
The terms $\sim dz d\bar u$ would give $\sim dz d u$ after going from $J$ to $I$ and are not allowed.

Further, the torsion components
\be
\lb{C-abc}
C_{ab\bar c} = \partial_b {\mathrm h}_{a \bar c} - \partial_a {\mathrm h}_{b \bar c}, \ \ \ \ \ \ \
C_{\bar a \bar b c} = \bar \partial_b {\mathrm h}_{c \bar a} - \bar\partial_a {\mathrm h}_{c \bar b}\, ,
\ee
which, in contrast to the mixed components \p{C-mixed}, do {\it not} change sign when going from $J$ to $I$, should vanish. The same is true for the mirror sector. And this implies
 \be
{\mathrm h}_{a \bar b}  \ =\ \partial_a \bar \partial_{\bar b} {\cal K}_1, \ \ \ \ \ \ \ \ \ \
 {\mathrm h}_{\alpha \bar \beta}  \ =\ \partial_\alpha \bar \partial_{\bar\beta} {\cal K}_2 \, .
 \ee
\end{proof}

One can suggest now the following geometric definition,

\begin{defi} An ordinary bi-K\"ahler manifold is a manifold with two commuting complex structures whose Bismut torsions are opposite and whose block diagonal metric \p{K-12} satisfies the additional constraint ${\cal K}_2 = - {\cal K}_1$.
 \end{defi}

If ${\cal K}_2 \neq -{\cal K}_1$, there is no reason to expect that the model admits two pairs of complex supercharges
satisfying the ${\cal N} = 4$ supersymmetry algebra.

In a generic quasicomplex bi-Kahler model \p{moreKah2} with the bosonic Lagrangian \p{lagr242-b-Ka-2}, the integration over the auxiliary fields modifies the metric as in \p{bez-A}. The Bismut torsions of the modified metric {\it have} the components like in \p{C-abc}, and
the property $C(I) = -C(J)$ does not hold anymore.

\section{Summary and outlook.}
We list again here the most essential original observations made in our paper.

\begin{enumerate}

\item We considered a generic ${\cal N} = 4$ Lagrangian involving a certain number of ordinary
({\bf 4}, {\bf 4}, {\bf 0}) multiplets \p{V-4-2exp} and a certain number of mirror multiplets \p{W-mu-expan} and  showed that it describes the {\it Clifford} KT supersymmetric quantum mechanical sigma models introduced in
\cite{CKT,Hull}.  Such models are characterized by the HKT geometry in both the ordinary and the mirror sector and can thus be called {\it bi-HKT} models.
They involve two different hypercomplex structures \p{comp-str-B-c}.

We presented explicit expressions for the superfields actions [Eqs.\p{act4} - \p{Del2a-VW}], the component
Lagrangian [Eqs.\p{lagr4-b-app}, \p{lagr4-2f-app},  \p{lagr4-4f-app}], the complex supercharges [ Eqs.\p{S-CKT}, \p{R-CKT}] and the real
supercharges \p{Q_real}. The latter allowed us to unravel the geometric structure of bi-HKT models (see {\bf Definition 6} and its justification in {\bf Theorem 1}).

\item  We performed the Hamiltonian reduction of a bi-HKT model with respect to the imaginary parts of all complex coordinates. The reduced model is described  via a certain number of ordinary and mirror   ({\bf 2}, {\bf 4}, {\bf 2}) multiplets [Eqs.\p{242-sf} - \p{242-const1}].  The model represents a generalization of the twisted K\"ahler model of Ref.\cite{GHR},
but involves extra holomorphic terms in the superfield action [see Eq. \p{moreKah2}].

The model can be described in terms of real ${\cal N} = 2$ superfields. It belongs to the class of quasicomplex
de Rham models with Hermitian (rather than just real) superfield metric \cite{quasi} and involves also {\it holomorphic torsions} such that the fermion charge is not conserved  \cite{torsion,FIS1}.  The model involves two different complex structures.  It is natural to call such models {\it bi-K\"ahler} models. Two complex
structures of a bi-K\"ahler model trace back their origin to two different hypercomplex structures of the parent bi-HKT model.

For a restricted class of these models given by \p{act242}, we derived the component Lagrangian [Eqs. \p{lagr242-b-Ka-2}, \p{lagr242-2f-add}, \p{lagr242-4f-add}]
and the supercharges [Eqs. \p{SR-HKT-red}, \p{SR-HKT-red-1}].

\end{enumerate}

In our paper, we studied the most general models described by a set of ordinary and mirror  linear ({\bf 4}, {\bf 4}, {\bf 0}) multiplets
or  ({\bf 2}, {\bf 4}, {\bf 2}) multiplets. But one could also consider  the models described by nonlinear multiplets.
We conjectured that such models describe again bi-HKT and bi-K\"ahler manifolds, but with nonzero Obata connections.

 Our second conjecture was that any hypercomplex manifold (a manifold possesing three quaternionic complex structures) is either HKT or bi-HKT.
It would be interesting to prove (or disprove) these conjectures.

 Another perspective direction of research is to include the {\it gauge} ${\cal N} = 4$ multiplets \cite{DI-gauge}
and semi-dynamical spinning multiplets \cite{FIL-spin} and study
 the models involving their interaction  with ({\bf 4}, {\bf 4}, {\bf 0}) or  ({\bf 2}, {\bf 4}, {\bf 2})  ``matter''
multiplets.   As was shown in \cite{Maxim}, one can obtain in this case the models involving gauge fields living on the manifold. Only the simplest cases with the flat metric or the conformally flat 4-dimensional metric were studied so far.
A generalization of this study to generic HKT  and bi-HKT manifolds would represent an interest.

\section*{Acknowledgements}

\noindent
We acknowledge the  support from the collaboration grant 03-58 of the  IN2P3-JINR program.
S.F.  acknowledge support from the RFBR grant 16-52-12012.
S.F.  thanks SUBATECH,
Universit\'{e} de Nantes,  and A.S. thanks JINR, Dubna for  warm hospitality
in the course of this study. We are indebted to E. Ivanov and M. Verbitsky for useful discussions.

\section*{Appendix A. Component Lagrangians.}
\setcounter{equation}0
\def\theequation{A.\arabic{equation}}

\subsection*{ A.1 \ bi-HKT}

The Lagrangian of a bi-HKT model involves besides the bosonic term
\p{lagr4-b-app} also the terms \p{lagr4-2f-app} \p{lagr4-4f-app}
\begin{eqnarray}
   L^{\rm bi-HKT}_{2f} &=& { \frac{i}{2}}\,h^{ab}_{m\bar n}
\left(\psi_a^m\dot{\bar\psi}_b^{\bar n}-\dot\psi_a^m\bar\psi_b^{\bar n} \right)+
{ \frac{i}{2}}\,h^{\alpha\beta}_{\mu\bar \nu}
\left(\chi_\alpha^\mu\dot{\bar\chi}_\beta^{\bar\nu}-\dot\chi_\alpha^\mu\bar\chi_\beta^{\bar\nu} \right)
\lb{lagr4-2f-app}\\ [11pt]
&& + i\left(\partial^c_k h_{m\bar n}^{ab}\right)\dot v_a^m\bar\psi_b^{\bar n}\psi_c^k
-i\left(\bar\partial^{c}_{\bar k} h_{m\bar n}^{ab}\right)\bar\psi_c^{\bar k}\psi_a^m \dot {\bar v}_b^{\bar n}
 + i\left(\partial^\gamma_\lambda h_{\mu\bar \nu}^{\alpha\beta}\right)\dot w_\alpha^\mu\bar\chi_\beta^{\bar \nu}\chi_\gamma^\lambda
-i\left(\bar\partial^{\gamma}_{\bar \lambda} h_{\mu\bar \nu}^{\alpha\beta}\right)\bar\chi_\gamma^{\bar \lambda}\chi_\alpha^\mu \dot {\bar w}_\beta^{\bar \nu}
\nonumber\\ [11pt]
&& +{ \frac{i}{2}}\left(\dot v_c^k\partial^{c}_k -\dot{\bar v}_c^{\bar k}\bar\partial^{c}_{\bar k}
+\dot w_\gamma^\lambda\partial^{\gamma}_\lambda -\dot{\bar w}_\gamma^{\bar \lambda}\bar\partial^{\gamma}_{\bar \lambda}\right)
\left( h^{ab}_{m\bar n}\,\psi_a^m\bar\psi_b^{\bar n} +h^{\alpha\beta}_{\mu\bar \nu}\,\chi_\alpha^\mu\bar\chi_\beta^{\bar \nu}\right)
\nonumber\\ [11pt]
&& + i\left(\partial^a_m h_{\mu\bar \nu}^{\alpha\beta}\right)\dot w_\alpha^\mu\bar\chi_\beta^{\bar \nu}\psi_a^m
-i\left(\bar\partial^{\alpha}_{\bar \mu} h_{m\bar n}^{ab}\right)\bar\chi_\alpha^{\bar \mu}\psi_a^m \dot {\bar v}_b^{\bar n}
+ i\left(\partial^\alpha_\mu h_{m\bar n}^{ab}\right)\dot v_a^m\bar\psi_b^{\bar n}\chi_\alpha^\mu
-i\left(\bar\partial^{a}_{\bar m} h_{\mu\bar \nu}^{\alpha\beta}\right)\bar\psi_a^{\bar m}\chi_\alpha^\mu \dot {\bar w}_\beta^{\bar \nu}
\nonumber\\ [11pt]
&& - { \frac{i}{2}}\,\dot w_\alpha^\mu\epsilon_{\mu\nu}\bar\partial^{\alpha}_{\bar \nu}\,
\epsilon_{nk}h^{ab}_{m\bar k}\,\psi_a^m\psi_b^{n}
+{ \frac{i}{2}} \, \dot{\bar w}_\alpha^{\bar \mu}\epsilon_{\bar \mu\bar \nu}\partial^{\alpha}_\nu \,
\epsilon_{\bar m\bar k} h^{ab}_{k\bar n}\,\bar\psi_a^{\bar m}\bar\psi_b^{\bar n}
\nonumber\\ [11pt]
&& -{ \frac{i}{2}}\,\dot v_a^m\epsilon_{mn}\bar\partial^{a}_{\bar n}\,
\epsilon_{\nu\lambda}h^{\alpha\beta}_{\mu\bar \lambda}\,\chi_\alpha^\mu\chi_\beta^{\nu}
+{ \frac{i}{2}} \, \dot{\bar v}_a^{\bar m}\epsilon_{\bar m\bar n}\partial^{a}_n\,
\epsilon_{\bar \mu\bar \lambda} h^{\alpha\beta}_{\lambda\bar \nu}\,\bar\chi_\alpha^{\bar \mu}\bar\chi_\beta^{\bar \nu}
\nonumber\\ [11pt]
&& +i\,\chi_\alpha^\mu\epsilon_{\mu\nu}\bar\partial^{\alpha}_{\bar \nu}\,
\epsilon_{nk}h^{ab}_{m\bar k}\,\dot v_a^m\psi_b^{n}
+i \, \psi_a^{m}\epsilon_{mn}\bar\partial^{a}_{\bar n}\,
\epsilon_{\nu\lambda} h^{\alpha\beta}_{\mu\bar \lambda}\,\dot{w}_\alpha^{\mu}\chi_\beta^{\nu}
\nonumber\\ [11pt]
&& -i\,\bar\chi_\alpha^{\bar\mu}\epsilon_{\bar\mu\bar\nu}\partial^{\alpha}_{\nu}\,
\epsilon_{\bar m\bar k}h^{ab}_{k\bar n}\,\dot {\bar v}_a^{\bar m}\bar\psi_b^{\bar n}
-i \, \bar\psi_a^{\bar m}\epsilon_{\bar m\bar n}\partial^{a}_{n}\,
\epsilon_{\bar\mu\bar\lambda} h^{\alpha\beta}_{\lambda\bar \nu}\,\dot{\bar w}_\alpha^{\bar\mu}\bar\chi_\beta^{\bar\nu} \,,
\nonumber
\end{eqnarray}
\begin{eqnarray}
\lb{lagr4-4f-app}
L^{\rm bi-HKT}_{4f} &=& -{ \frac{1}{4}}\,\epsilon_{np}\,\epsilon_{\bar k \bar r}
\left(\Delta^{cd}_{r\bar l}h^{ab}_{ m\bar p}
\right)\psi_a^{m}\psi_b^{n}\bar\psi_c^{\bar k}\bar\psi_d^{\bar l}
+{ \frac{1}{4}}\,\epsilon_{\nu\sigma}\,\epsilon_{\bar \lambda\bar \kappa}
\left(\Delta^{\gamma\delta}_{\kappa\bar \rho}h^{\alpha\beta}_{\mu\bar \sigma}
\right)\chi_\alpha^{\mu}\chi_\beta^{\nu}\bar\chi_\gamma^{\bar \lambda}\bar\chi_\delta^{\bar \rho}\\ [11pt]
&& +{ \frac{1}{2}}\,\bar\chi_\alpha^{\bar \mu}\bar\partial^{\alpha}_{\bar \mu}\,\bar\psi_c^{\bar k}\epsilon_{\bar k \bar l}\partial^{c}_{l}\,
\epsilon_{np}h^{ab}_{ m\bar p}\psi_a^{m}\psi_b^{n}
-{ \frac{1}{2}}\,\chi_\alpha^{\mu}\partial^{\alpha}_{\mu}\,\psi_c^{k}\epsilon_{k l}\bar\partial^{c}_{\bar l}\,
\epsilon_{\bar m\bar p}h^{ab}_{ p\bar n}\bar \psi_a^{\bar m}\bar \psi_b^{\bar n}
\nonumber\\ [11pt]
&& +{ \frac{1}{2}}\,\bar\psi_a^{\bar m}\bar\partial^{a}_{\bar m}\,\bar\chi_\gamma^{\bar \lambda}\epsilon_{\bar \lambda \bar \rho}\partial^{\gamma}_{\rho}\,
\epsilon_{\nu \sigma}h^{\alpha\beta}_{ \mu\bar \sigma}\chi_\alpha^{\mu}\chi_\beta^{\nu}
-{ \frac{1}{2}}\,\psi_a^{m}\partial^{a}_{m}\,\chi_\gamma^{\lambda}\epsilon_{\lambda\rho}\bar\partial^{\gamma}_{\bar\rho}\,
\epsilon_{\bar\mu \bar\sigma}h^{\alpha\beta}_{ \sigma\bar \nu}\bar\chi_\alpha^{\bar\mu}\bar\chi_\beta^{\bar\nu}
\nonumber\\ [11pt]
&&
-{ \frac{1}{2}}\,\chi_\alpha^{\mu}\bar\chi_\beta^{\bar \nu}
\left(\partial^{\alpha}_{\mu}\bar\partial^{\beta}_{\bar\nu}-
\epsilon_{\mu\lambda}\epsilon_{\bar\nu\bar\rho}\bar\partial^{\alpha}_{\bar\lambda}\partial^{\beta}_{\rho}\right)
h^{ab}_{ m\bar n}\psi_a^{m}\bar \psi_b^{\bar n}
-{ \frac{1}{2}}\,\psi_a^{m}\bar\psi_b^{\bar n}
\left(\partial^{a}_{m}\bar\partial^{b}_{\bar n}-
\epsilon_{mk}\epsilon_{\bar n\bar l}\bar\partial^{a}_{\bar k}\partial^{b}_{l}\right)
h^{\alpha\beta}_{ \mu\bar \nu}\chi_\alpha^{\mu}\bar \chi_\beta^{\bar \nu}
\nonumber\\ [11pt]
&& -{ \frac{1}{2}}\left(\bar\psi_c^{\bar l}\bar\partial^{c}_{\bar l}+\bar\chi_\gamma^{\bar \lambda}\bar\partial^{\gamma}_{\bar\lambda}\right)
\left(\chi_\alpha^{\mu}\epsilon_{\mu\nu}\bar\partial^{\alpha}_{\bar\nu}\,
\epsilon_{np}h^{ab}_{ m\bar p}\psi_a^{m}\psi_b^{n}
+\psi_a^{m}\epsilon_{mn}\bar\partial^{a}_{\bar n}\,
\epsilon_{\nu\sigma}h^{\alpha\beta}_{\mu\bar\sigma}\chi_\alpha^{\mu}\chi_\beta^{\nu}\right)
\nonumber\\ [11pt]
&&
+{ \frac{1}{2}}\left(\psi_c^{l}\partial^{c}_{l}+\chi_\gamma^{\lambda}\partial^{\gamma}_{\lambda}\right)
\left(\bar\chi_\alpha^{\bar\mu}\epsilon_{\bar\mu\bar\nu}\partial^{\alpha}_{\nu}\,
\epsilon_{\bar m\bar p}h^{ab}_{p\bar n}\bar\psi_a^{\bar m}\bar\psi_b^{\bar n}+
\bar\psi_a^{\bar m}\epsilon_{\bar m\bar n}\partial^{a}_{n}\,
\epsilon_{\bar\mu\bar\sigma}h^{\alpha\beta}_{\sigma\bar\nu}\bar\chi_\alpha^{\bar\mu}\bar\chi_\beta^{\bar\nu}
\right),
\nonumber
  \end{eqnarray}

\subsection*{A.2\ bi-K\"ahler}

We present here the full component expression for the bi-K\"ahler Lagrangian chosen
in the restricted form \p{act242} and expressed via the chiral ${\cal N} = 2$ superfields as in \p{act422-2}.  Its bosonic part was written in \p{lagr242-b-Ka-2}.
There are also 2-fermion and 4-fermion terms,
\be
\lb{lagr242-2f-add}
\begin{array}{rcl}
 L^{\rm bi-K}_{2f} &=& {\displaystyle\frac{i}{2}}\,\kappa_{a\bar b}
\left(\phi^a\dot{\bar\phi}^{\bar b}-\dot\phi^a\bar\phi^{\bar b} +\varphi^a\dot{\bar\varphi}^{\bar b}
-\dot\varphi^a\bar\varphi^{\bar b} \right)+{\displaystyle\frac{i}{2}}\,\kappa_{\alpha\bar \beta}
\left(\rho^\alpha\dot{\bar\rho}^{\bar \beta}-\dot\rho^\alpha\bar\rho^{\bar \beta}+
\varrho^\alpha\dot{\bar\varrho}^{\bar \beta}-\dot\varrho^\alpha\bar\varrho^{\bar \beta}\right)
\\ [11pt]
&& +{\displaystyle \frac12}
\,{\cal F}_{[ab]}\left(\varphi^{a}\dot\phi^{b}-\phi^{a}\dot\varphi^{b} \right)-
{\displaystyle \frac12}\,
\bar{\cal F}_{[\bar a\bar b]}
\left(\bar\varphi^{\bar a}\dot{\bar\phi}^{\bar b}-\bar\phi^{\bar a}\dot{\bar\varphi}^{\bar b} \right)
\\ [11pt]
&& +{\displaystyle \frac12}\,
{\cal G}_{[\alpha\beta]} \left(\varrho^{\alpha}\dot\rho^{\beta}-\rho^{\alpha}\dot\varrho^{\beta} \right)-
{\displaystyle \frac12}\,
\bar{\cal G}_{[\bar\alpha\bar\beta]}
\left(\bar\varrho^{\bar\alpha}\dot{\bar\rho}^{\bar\beta}-\bar\rho^{\bar\alpha}\dot{\bar\varrho}^{\bar\beta} \right)
\\ [11pt]
&& - {\displaystyle\frac{i}{2}}\left(\dot z^c\partial_{c}-\dot {\bar z}^{\bar c}\bar\partial_{\bar c}\right)
\kappa_{a\bar b}\left(\phi^a\bar\phi^{\bar b}+ \varphi^a\bar\varphi^{\bar b}\right)
-{\displaystyle\frac{i}{2}}\left(\dot u^\gamma\partial_{\gamma}-\dot {\bar u}^{\bar\gamma}\bar\partial_{\bar\gamma}\right)
\kappa_{\alpha\bar \beta}\left(\rho^\alpha\bar\rho^{\bar \beta}+ \varrho^\alpha\bar\varrho^{\bar \beta}\right)
\\ [11pt]
&& + {\displaystyle\frac{i}{2}}\left(\dot u^\gamma\partial_{\gamma}-\dot {\bar u}^{\bar\gamma}\bar\partial_{\bar\gamma}\right)
\kappa_{a\bar b}\left(\phi^a\bar\phi^{\bar b}- \varphi^a\bar\varphi^{\bar b}\right)
+ {\displaystyle\frac{i}{2}}\left(\dot z^c\partial_{c}-\dot {\bar z}^{\bar c}\bar\partial_{\bar c}\right)
\kappa_{\alpha\bar \beta}\left(\rho^\alpha\bar\rho^{\bar \beta}- \varrho^\alpha\bar\varrho^{\bar \beta}\right)
\\ [11pt]
&& -i\left(\partial_{\gamma}\kappa_{a\bar b}\right)
\rho^\gamma\dot z^{a}\bar\phi^{\bar b}
- i\left(\bar\partial_{\bar\gamma}\kappa_{a\bar b}\right)
\bar\rho^{\bar\gamma}\phi^{a}\dot{\bar z}^{\bar b}
-i\left(\bar\partial_{\bar\gamma}\kappa_{a\bar b}\right)
\bar\varrho^{\bar\gamma}\dot z^{a}\bar\varphi^{\bar b}
- i\left(\partial_{\gamma}\kappa_{a\bar b}\right)
\varrho^{\gamma}\varphi^{a}\dot{\bar z}^{\bar b}
\\ [11pt]
&& -i\left(\partial_{c}\kappa_{\alpha\bar \beta}\right)
\phi^c\dot u^{\alpha}\bar\rho^{\bar \beta}
- i\left(\bar\partial_{\bar c}\kappa_{\alpha\bar \beta}\right)
\bar\phi^{\bar c}\rho^{\alpha}\dot{\bar u}^{\bar \beta}
-i\left(\bar\partial_{\bar c}\kappa_{\alpha\bar \beta}\right)
\bar\varphi^{\bar c}\dot u^{\alpha}\bar\varrho^{\bar \beta}
- i\left(\partial_{c}\kappa_{\alpha\bar \beta}\right)
\varphi^{c}\varrho^{\alpha}\dot{\bar u}^{\bar \beta}
\\ [11pt]
&& -\left(\bar\partial_{\bar c}\kappa_{a\bar b}\right)
\bar\phi^{\bar c}A^a\bar\varphi^{\bar b}
+ \left(\partial_{c}\kappa_{a\bar b}\right)
\phi^{c}\varphi^a{\bar A}^{\bar b}
-\left(\bar\partial_{\bar\gamma}\kappa_{\alpha\bar \beta}\right)
\bar\rho^{\bar \gamma}B^\alpha\bar\varrho^{\bar \beta}
+ \left(\partial_{\gamma}\kappa_{\alpha\bar \beta}\right)
\rho^{\gamma}\varrho^\alpha{\bar B}^{\bar \beta}
\\ [11pt]
&& -\left(\bar\partial_{\bar \gamma}\kappa_{a\bar b}\right)
\bar\rho^{\bar \gamma}A^a\bar\varphi^{\bar b}
+ \left(\partial_{\gamma}\kappa_{a\bar b}\right)
\rho^{\gamma}\varphi^a{\bar A}^{\bar b}
-\left(\bar\partial_{\bar c}\kappa_{\alpha\bar \beta}\right)
\bar\phi^{\bar c}B^\alpha\bar\varrho^{\bar \beta}
+ \left(\partial_{c}\kappa_{\alpha\bar \beta}\right)
\phi^{c}\varrho^\alpha{\bar B}^{\bar \beta}
\\ [11pt]
&& +\left(\partial_{\gamma}\kappa_{a\bar b}\right)
\varrho^{\gamma}A^a\bar\phi^{\bar b}
- \left(\bar\partial_{\bar\gamma}\kappa_{a\bar b}\right)
\bar\varrho^{\bar\gamma}\phi^a{\bar A}^{\bar b}
+\left(\partial_{c}\kappa_{\alpha\bar \beta}\right)
\varphi^{c}B^\alpha\bar\rho^{\bar \beta}
- \left(\bar\partial_{\bar c}\kappa_{\alpha\bar \beta}\right)
\bar\varphi^{\bar c}\rho^\alpha{\bar B}^{\bar \beta}
\\ [11pt]
&& -A^c\left(\partial_{c}\kappa_{\alpha\bar \beta}\right)
\varrho^{\alpha}\bar\rho^{\bar \beta}
- {\bar A}^{\bar c}\left(\bar\partial_{\bar c}\kappa_{\alpha\bar \beta}\right)
\rho^\alpha\bar\varrho^{\bar\beta}
-B^\gamma\left(\partial_{\gamma}\kappa_{a\bar b}\right)
\varphi^{a}\bar\phi^{\bar b}
- {\bar B}^{\bar \gamma}\left(\bar\partial_{\bar \gamma}\kappa_{a\bar b}\right)
\phi^a\bar\varphi^{\bar b}
\\ [11pt]
&& +{\displaystyle \frac12}\,\phi^{c}(\partial_{c}{\cal F}_{[ab]})\dot z^{a}\varphi^{b}
-{\displaystyle \frac12}\,\varphi^{c}(\partial_{c}{\cal F}_{[ab]})\dot z^{a}\phi^{b}
-{\displaystyle \frac12}\,\bar\phi^{\bar c}(\bar\partial_{\bar c}\bar{\cal F}_{[\bar a\bar b]})\dot {\bar z}^{\bar a}\bar\varphi^{\bar b}
+{\displaystyle \frac12}\,\bar\varphi^{\bar c}(\bar\partial_{\bar c}\bar{\cal F}_{[\bar a\bar b]})\dot {\bar z}^{\bar a}\bar\phi^{\bar b}
\\ [11pt]
&& +{\displaystyle \frac12}\,\rho^{\gamma}(\partial_{\gamma}{\cal G}_{[\alpha\beta]})\dot u^{\alpha}\varrho^{\beta}
-{\displaystyle \frac12}\,\varrho^{\gamma}(\partial_{\gamma}{\cal G}_{[\alpha\beta]})\dot u^{\alpha}\rho^{\beta}
-{\displaystyle \frac12}\,\bar\rho^{\bar\gamma}(\bar\partial_{\bar\gamma}
\bar{\cal G}_{[\bar\alpha\bar\beta]})\dot {\bar u}^{\bar\alpha}\bar\varrho^{\bar\beta}
+{\displaystyle \frac12}\,\bar\varrho^{\bar\gamma}(\bar\partial_{\bar\gamma}
\bar{\cal G}_{[\bar\alpha\bar\beta]})\dot {\bar u}^{\bar\alpha}\bar\rho^{\bar\beta}\,,
\end{array}
\ee
\be
\lb{lagr242-4f-add}
\begin{array}{rcl}
L^{\rm bi-K}_{4f} &=& \left(\partial_{c}\bar\partial_{\bar d}\kappa_{a\bar b} \right)
\phi^c\bar\phi^{\bar d}\varphi^a\bar\varphi^{\bar b}+
\left(\partial_{\gamma}\bar\partial_{\bar\delta}\kappa_{\alpha\bar \beta} \right)
\rho^\gamma\bar\rho^{\bar\delta}\varrho^\alpha\bar\varrho^{\bar\beta}
\\ [11pt]
&&+\left(\partial_{\gamma}\bar\partial_{\bar\delta}\kappa_{a\bar b} \right)
\rho^\gamma\bar\rho^{\bar\delta}\varphi^a\bar\varphi^{\bar b}-\left(\partial_{\beta}\partial_{\gamma}\kappa_{a\bar b} \right)
\rho^\beta\varrho^{\gamma}\varphi^a\bar\phi^{\bar b}-
\left(\bar\partial_{\bar \beta}\bar\partial_{\bar \gamma}\kappa_{a\bar b} \right)
\bar\varrho^{\bar \gamma}\bar\rho^{\bar \beta}\phi^a\bar\varphi^{\bar b}
\\ [11pt]
&&+
\left(\partial_{c}\bar\partial_{\bar d}\kappa_{\alpha\bar \beta} \right)
\phi^c\bar\phi^{\bar d}\varrho^\alpha\bar\varrho^{\bar\beta}
-\left(\partial_{b}\partial_{c}\kappa_{\alpha\bar \beta} \right)
\phi^{b}\varphi^{c}\varrho^\alpha\bar\rho^{\bar\beta}-
\left(\bar\partial_{\bar b}\bar\partial_{\bar c}\kappa_{\alpha\bar \beta} \right)
\bar\varphi^{\bar c}\bar\phi^{\bar b}\rho^\alpha\bar\varrho^{\bar\beta}
\\ [11pt]
&&+\left(\partial_{c}\bar\partial_{\bar\gamma}\kappa_{a\bar b}\right)
\phi^c\bar\rho^{\bar\gamma}\varphi^a\bar\varphi^{\bar b}+
\left(\partial_{\gamma}\bar\partial_{\bar c}\kappa_{a\bar b} \right)
\rho^\gamma\bar\phi^{\bar c}\varphi^a\bar\varphi^{\bar b}
\\ [11pt]
&&+\left(\partial_{c}\partial_{\gamma}\kappa_{a\bar b} \right)
\varphi^c\varrho^{\gamma}\phi^a\bar\phi^{\bar b}+
\left(\bar\partial_{\bar c}\bar\partial_{\bar \gamma}\kappa_{a\bar b} \right)
\bar\varrho^{\bar \gamma}\bar\varphi^{\bar c}\phi^a\bar\phi^{\bar b}
\\ [11pt]
&&+\left(\partial_{c}\bar\partial_{\bar\gamma}\kappa_{\alpha\bar \beta} \right)
\phi^c\bar\rho^{\bar\gamma}\varrho^\alpha\bar\varrho^{\bar\beta}+
\left(\partial_{\gamma}\bar\partial_{\bar c}\kappa_{\alpha\bar \beta} \right)
\rho^\gamma\bar\phi^{\bar c}\varrho^\alpha\bar\varrho^{\bar\beta}
\\ [11pt]
&&-\left(\partial_{c}\partial_{\gamma}\kappa_{\alpha\bar \beta} \right)
\varphi^{c}\varrho^{\gamma}\rho^\alpha\bar\rho^{\bar\beta}-
\left(\bar\partial_{\bar c}\bar\partial_{\bar\gamma}\kappa_{\alpha\bar \beta} \right)
\bar\varrho^{\bar\gamma}\bar\varphi^{\bar c}\rho^\alpha\bar\rho^{\bar\beta}
\, ,
\end{array}
\ee
where $\kappa_{a\bar b}, {\cal F}_{ab}$ and  $\kappa_{\alpha \bar \beta}, {\cal G}_{\alpha\beta}$
were defined in \p{kap-FG}.

\end{document}